\documentclass[11pt]{article}

\usepackage{amsmath}

\usepackage{times}
\usepackage{bm}
\usepackage{natbib}
\usepackage{amssymb}
\usepackage{graphicx}
\usepackage[margin=1.3in]{geometry}
\usepackage[plain,noend]{algorithm2e}
\usepackage{paralist}
\usepackage{amsthm}
\usepackage{authblk}
\usepackage{color}
\usepackage{hyperref}
\usepackage{comment}

\usepackage{multirow}

\hypersetup{
    colorlinks=true,
    linkcolor=blue,
    filecolor=blue,      
    urlcolor= blue,
    citecolor=blue,
}

\newtheorem{theorem}{Theorem}[section]
\newtheorem{lemma}[]{Lemma}[section]
\newtheorem{proposition}[]{Proposition}[section]
\newtheorem{remark}[]{Remark}[section]
\newtheorem{definition}[]{Definition}[section]
\newtheorem{assumption}[]{Assumption}[section]

\DeclareMathOperator*{\argsup}{arg\,sup}

\DeclareMathOperator{\logit}{\text{logit}}

\usepackage{fancyhdr}
\begin{document}


\title{Integer-valued autoregressive models with   survival probability driven by a stochastic recurrence equation \footnotemark \footnotetext{ Email address: \texttt{gorgi@stat.unipd.it}
}}

\author[]{P. Gorgi}

\affil[]{University of Padua, Italy\\VU University Amsterdam, The Netherlands}

{\let\newpage\relax\maketitle}

\begin{abstract}
A new class of integer-valued autoregressive  models with dynamic survival probability is proposed. The peculiarity of this class of models lies on the specification of the survival probability through a stochastic recurrence equation. The estimation of the model can be performed by maximum likelihood and the consistency of the estimator is proved in a misspecified model setting.  The flexibility of the proposed  specification is illustrated in a simulation study. An application to a  time series of crime reports is presented. The results show how the   dynamic survival probability  can  enhance both   in-sample and  out-of-sample performances of integer-valued autoregressive models. 
\end{abstract}




\section{Introduction} 

Over the last few years, there has been an increasing interest in modeling and forecasting integer-valued time series. The reason being that  many  observed time series are not continuous and the use of specific models  to take this into account allows us to better describe  time series behaviors.  One of the most popular models for  time series of counts is the  Integer-valued Autoregressive (INAR) model of \cite{ALOSH1987} and \cite{MCK1988}. Its  specification is based on the thinning operator `$\circ$'  of \cite{steutel1979discrete}. For a given $N \in \mathbb{N}$ and $\alpha \in (0,1)$, the thinning operator is  defined to satisfy $\alpha \circ N =\sum_{i=1}^{N}x_i$, where $\{x_i\}_{i=1}^N$ is a sequence of  independent Bernulli random variables with success probability $\alpha$. 
The thinning operator  enables the specification of integer-valued time series models in an autoregressive fashion. In fact,  INAR models can be seen as a discrete response version of the well known linear autoregressive  model. The  first order  INAR model is described by the following equation
\begin{equation}\label{inar}
y_t=\alpha \circ y_{t-1} + \varepsilon_t, \; t \in \mathbb{Z},
 \end{equation}
 where $\{\varepsilon_t\}_{t \in \mathbb{Z}}$ is an i.i.d.~sequence of integer-valued random variables. An appealing feature of the INAR model in (\ref{inar}) is its well known interpretation as a death-birth process. From this interpretation, the coefficient $\alpha$ is also  called the survival probability.  As  in the original formulation of \cite{ALOSH1987} and \cite{MCK1988}, the error term $\varepsilon_t$ is typically assumed to be   Poisson distributed.   Other distributions have also been considered in the literature   as the Poisson imposes equidispersion and this is can be restrictive in practice, see \cite{al1992first} and \cite{jazi2012}. Besides the distribution of the error term, the  INAR specification in (\ref{inar})  has been generalized in several directions. Among others, \cite{alzaid1990} and \cite{jin1991} extended the first order INAR model to a general order $p$, \cite{kim2008non} considered a signed thinning operator to handle nonstationary series and \cite{pedeli2011bivariate} introduced a bivariate INAR model.



Real time series data often exhibit changing dynamic behaviors.  As a result, employing more flexible  specifications for the dynamic component of the model  can provide a better description of  the underlying behavior of the time series  and  produce better forecasts. The contribution of this paper is in this direction: we introduce a new class of INAR models with time-varying survival probability. The peculiarity of our approach is that the dynamics of the INAR coefficient is specified through a Stochastic Recurrence Equation (SRE) that is driven by the score of the predictive  likelihood. This method allows  the survival probability  to be updated at each time period using the information provided by past elements of the series. The use of the score to update time-varying parameters has been recently  proposed by \cite{Creal2013} and \cite{H2013}. Since then, their Generalized Autoregressive Score (GAS) framework has been successfully employed to develop dynamic models  in  econometrics and time series analysis, see for instance    \cite{salvatierra2015dynamic}, \cite{HL2014} and \cite{creal2012dynamic}.  It is also worth mentioning that  many well-known observation-driven models turn out to be GAS models. Examples include the   GARCH model of \cite{engle1982} and \cite{bol1986} and, in the context of integer-valued time series, the Poisson autoregressive model of \cite{Davis2003}.  For a more detailed discussion see \cite{Creal2013}.

In the literature, time variation of the INAR survival probability has also  been considered by    \cite{Zheng2007212} and \cite{zheng2008first}.  In  \cite{Zheng2007212}  the survival probability is specified  as  a sequence of i.i.d.~random variables. This approach leads to a more flexible class of conditional distributions  but,  because of the i.i.d.~assumption,  it does not provide a dynamic specification of the INAR coefficient.   \cite{zheng2008first} allows the INAR coefficient  to depend on past observations. Their method introduces a dynamic structure and the survival probability is updated using past information as in our approach. However, as we shall see in Section \ref{section3}, their  specification is not able  to properly  model  smooth changes  of the survival probability.

The  INAR model we propose  in this paper should not be interpreted as a Data Generating Processes (DGP) but as a filter to approximate the distribution of  a more complex and unknown DGP. The reasoning behind this interpretation is provided by  the work of \cite{blasq2015}. In particular, \cite{blasq2015} show that score-driven time-varying parameters should be employed in a misspecified model setting as they are optimal in reducing the Kullback-Leibler (KL) divergence with respect to an  unknown  true DGP. In this direction, we  illustrate the flexibility of the proposed dynamic specification for the INAR coefficient  by means of a simulation study in a misspecified framework. The results illustrate how the model is able to capture complex dynamic behaviors and well approximate the true distribution of different  DGPs.  Furthermore, we derive some statistical properties of the Maximum Likelihood (ML) estimator: we prove the its consistency  in a misspecified setting and  show that also the conditional predictive probability mass function (pmf) can be consistently estimated through a  plug-in estimator. In particular, the plug-in pmf estimator is shown to converge to a pseudo-true conditional pmf that has the interpretation of minimizing on average the KL divergence with  the true conditional pmf of the DGP. These results are useful not only to ensure    the reliability of inference but also forecasting. Finally, the practical usefulness of the  proposed model  is illustrated thorough  an application to a  real time series dataset of   crime reports. The results are promising and show how the  dynamic survival probability can   enhance both in-sample and  out-of-sample performances of INAR models.

The paper is structured as follows. Section \ref{section1} introduces the class of models. Section \ref{section2} discusses the consistency of ML estimation. Section \ref{section3}  presents the Monte Carlo simulation experiments. Section \ref{section4}  illustrates the empirical application.  Section \ref{section5} concludes.

\section{INAR models with score-driven coefficient} 
\label{section1}

\subsection{The class of models} 
In this section, we extend the class of  INAR models  in (\ref{inar})  by  allowing the survival probability $\alpha$ to change over time. The dynamics  of the time-varying coefficient  $\alpha_t$ is specified on the basis of the score framework of \cite{Creal2013} and \cite{H2013}. The GAS-INAR  model is described by the following equations
\begin{align}
y_t=& \alpha_t \circ y_{t-1} + \varepsilon_t, \label{m1}\\
  \logit\alpha_{t+1} =&\omega+\beta \logit\alpha_{t} + \tau s_t,
\label{m2}
 \end{align}
where  $\{\varepsilon_t\}_{t\in\mathbb{Z}}$ is an i.i.d.~sequence of random variables with  pmf $p_e(x,\xi)$ for $x \in \mathbb{N}$, $\xi \in \Xi\subseteq\mathbb{R}^k$,   the vector $\theta=(\omega, \beta, \tau, \xi)^T$ is a $k+3$ dimensional parameter vector to be estimated  and $s_t=s_t(\alpha_{t},\xi)$  denotes the score of the predictive log-likelihood $\partial \log p(y_t|\alpha_t, y_{t-1}, \xi)/\partial \logit\alpha_t$. Note that throughout the paper we consider the convention that the set $\mathbb{N}$  includes also zero. The functional form of the predictive likelihood  $p(y_t|\alpha_t, y_{t-1}, \xi)$ can be obtained as the convolution between the conditional pmf of $\alpha_t \circ y_{t-1}$ and the pmf of the error term $\varepsilon_t$, i.e.
$$p(y_t|\alpha_t, y_{t-1}, \xi)=\sum_{k=0}^{\min\{y_t, y_{t-1}\}} p_b(k, y_{t-1}, \alpha_t) p_e(y_{t}-k,\xi),$$
where  $p_b(x, y_{t-1}, \alpha_t)$ for $x \in \{0,\dots, y_{t-1}\}$ is the pmf of a Binomial random variable with size $y_{t-1}$ and success probability $\alpha_t$. An analytical  expression of the score innovation $s_t$ can be found in  Appendix \ref{appendixA1}. The  logit link function in the SRE in  (\ref{m2}) is considered  to  ensure that the survival probability $\alpha_t$ is between zero and one.

The GAS-INAR model  in (\ref{m1}) and (\ref{m2}) retains the  interpretation of INAR models as  death-birth processes.
 In particular, the observed number of elements $y_t$ alive at time $t$ is given by the sum between the number of surviving elements from time $t-1$ and the new birth elements $\varepsilon_t$.  In our dynamic specification, each of the elements alive at time $t-1$ has a probability $\alpha_t$ of surviving at time $t$.  We also note that the proposed model  is observation-driven as the dynamic probability $\alpha_t$ is driven solely by past observations.  The score $s_t$ can be seen as the innovation of the dynamic system in  (\ref{m2}) as it provides the new information that becomes available   at time $t$ observing $y_t$. The interpretation of $s_t$ as an  innovation is further   justified by the fact that its  conditional expectation $E(s_t|y_{t-1},\alpha_t)$ is equal to zero.

\begin{figure}[h!]
\center
\includegraphics[scale=0.65]{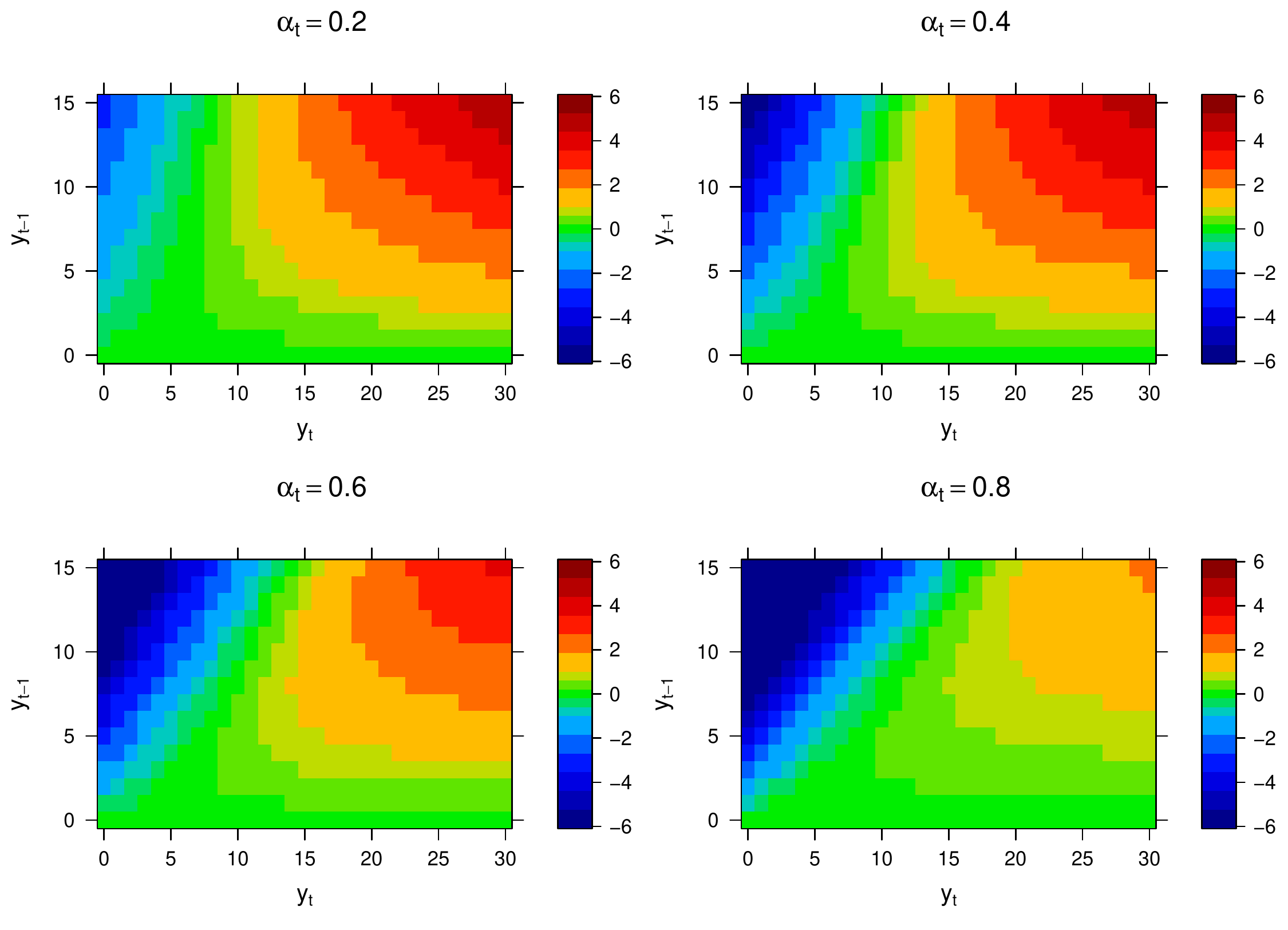}
\vspace{-0.25cm}
\caption{\textit{Impact of $y_t$ and $y_{t-1}$ on the score innovation $s_t$ for different values of the survival probability $\alpha_t$. A Poisson distribution with mean equal to five is considered as distribution of the error term $\varepsilon_t$. }}
\label{fig_imp}
\end{figure}

It is interesting to see how  the information obtained observing $y_t$ is processed through the score $s_t$ to update the survival probability from $\alpha_t$ to $\alpha_{t+1}$.  Figure \ref{fig_imp} describes the impact of $y_t$ on $s_t$ for different values of $y_{t-1}$ and $\alpha_t$. As we can see from the plots, the survival probability $\alpha_t$ gets a negative update when  $y_t$ is small and $y_{t-1}$ is large. This has an intuitive explanation: the information about $\alpha_t$ we get observing a small $y_t$ after a large $y_{t-1}$ is that the survival probability should be small as otherwise with a large $\alpha_t$ we wold expect many elements from time $t-1$ to survive and thus a large $y_t$ as well.  As a result,    $\alpha_t$  should get a negative update to discount this information.  Similarly, observing  a large $y_t$ following a large $y_{t-1}$ suggests an high survival probability. Thus, the probability $\alpha_t$ should be updated accordingly and get a positive innovation $s_t$. Finally, an innovation $s_t$ close to zero  may be an indication of  either  a lack of information or that the observed value of $y_t$ is compatible with the value $y_{t-1}$ and the current state of the  survival probability $\alpha_t$. The former case reflects  situations when $y_{t-1}$ is  zero (or close to zero). This because observing $y_t$ provides no information on the survival probability of the elements $y_{t-1}$ as there are no elements alive at $t-1$. On the other hand, the latter case of observing a value $y_t$ compatible with $y_{t-1}$ and $\alpha_t$ can be seen as the green area that separates the red an the blue areas in Figure \ref{fig_imp}. 

This line of reasoning concerning the direction of the update $s_t$ is subject to the current value of $\alpha_t$.  For instance, in a situation where $\alpha_t$ is  close to zero  perhaps observing a small $y_t$ after a large $y_{t-1}$ is exactly what we would expect. Thus the score update $s_t$ may be close to zero in this case.  This  dependence of the score update $s_t$  on the current  survival probability $\alpha_t$ can be noted across the different plots in Figure \ref{fig_imp}.

It is also worth mentioning that the functional form of the score innovation $s_t$  depends on the specification of the pmf of the error term $\varepsilon_t$ as the predictive likelihood depends on it. In practice, the  pmf $p_e(x,\xi)$ can be chosen in such a way  to take into account the main features observed in the data. For instance, as we will consider in the application in Section \ref{section4}, a Negative Binomial distribution may be considered instead of a Poisson when the data suggests overdispersion. Alternatively, a zero inflated Poisson or  Negative Binomial distributions may be employed when dealing with time series with a  large number of zeros.

\subsection{Parameter estimation}   

The static parameter vector $\theta$ of the  GAS-INAR model  can be estimated by ML. The log-likelihood function is available in closed form through a prediction error decomposition, namely
\begin{align*}
\hat L_T(\theta)=\frac{1}{T}\sum_{t=1}^T \log p(y_t|\hat \alpha_t(\theta), y_{t-1}, \xi).
\end{align*}
The filtered survival probability $\hat \alpha_t(\theta)$ is obtained recursively using the observed data $\{y_t\}_{t=0}^T$ as 
\begin{align}
  \logit\hat \alpha_{t+1}(\theta)=\omega+\beta \logit \hat \alpha_{t}(\theta) + \tau s_t( \hat \alpha_{t}(\theta), \xi),
  \label{filter}
 \end{align}
 where the recursion is initialized at a fixed point $\logit \hat \alpha_{0}(\theta)\in \mathbb{R}$. A reasonable choice for the initialization is $\logit \hat \alpha_{0}(\theta)=\omega/(1-\beta)$. That  is the unconditional mean  $E \text{logit} \alpha_{t}$ implied by the GAS-INAR model under the parametric assumption $\theta$. This follows immediately as the expected value of the score is equal to zero under standard regularity conditions. The ML estimator is finally given by
 \begin{align}
\hat \theta_T =\argsup_{\theta \in \Theta} \hat L_T(\theta),
\label{MLE}
\end{align}
where $\Theta$ is a compact parameter set contained in $\mathbb{R}\times (-1,1) \times \mathbb{R} \times \Xi$.
 
The asymptotic stability of the filtered parameter $\logit\hat\alpha_t(\theta)$ and the consistency of the ML estimator  as well as the predictive distribution are studied in Section \ref{section2}. Furthermore, in Section \ref{section3},   a simulation experiment is performed to study the finite sample behavior of the ML estimator and to further confirm its reliability.

\subsection{Forecasting} 

One of the  advantages of properly modeling count time series taking into account the discreteness of the data is that it is possible to obtain coherent forecasts of the entire pmf.  As shown in \cite{Freeland2004427}, forecasts $h$ steps ahead are typically available in closed form for INAR models as in  (\ref{inar}).  The  conditional pmf $h$ steps ahead  can be obtained by repeated applications of the convolution formula. Similarly, for point forecasts, a closed form expression is available as the conditional expectation $h$ steps ahead is $E(y_{T+h}|y_{T})= \alpha^h y_{T} + \mu$, with $\mu=E(\varepsilon_t)$. 

In the following, we illustrate a possible way to obtain $h$ steps ahead forecasts from the GAS-INAR model. A closed form expression for the conditional  pmf  $h$ steps ahead $p_{T+h}(x)$ is only available for  $h=1$. In particular, it is given by
$$p_{T+1}(x)=\sum_{k=0}^{\min\{x, y_T\}} p_{b}(k, y_T, \alpha_T)p_{e}(x-k).$$
Numerical methods are required to obtain $p_{T+h}(x)$ for $h\ge 2$.   A possibility is to approximate $p_{T+h}(x)$ considering the following simulation scheme.  First, simulate  $B$ realization for $y_{T+h}$,  $y_{T+h}^{(i)}$, $i=1,\dots, B$.  Then, obtain an approximation of $p_{T+h}(x)$ as $\hat p_{T+h}(x)=n_x^h/B$, where $n_x^h$ denotes the number  of draws $y_{T+h}^{(i)}$, $i=1,\dots, B$, equal to $x$.
The simulations of $y_{T+h}^{(i)}$, $i=1,\dots, B$,  can be  performed considering the following procedure.  For $k=1,\dots, h$
\begin{enumerate}
\item Simulate   $\varepsilon_k^{(i)}$ from the distribution   $p_e(x,\xi)$ and  $\alpha_{T+k}^{(i)}\circ y_{T+k-1}^{(i)}$ from a Binomial distribution with size $y_{T+k-1}^{(i)}$ and success probability $\alpha_{T+k}^{(i)}$.
\item Compute $y_{T+k}^{(i)}=\alpha_{T+k}^{(i)}\circ y_{T+k-1}^{(i)}+\varepsilon_k^{(i)}$ and  update $ \alpha_{T+k}^{(i)}$ to $ \alpha_{T+k+1}^{(i)}$ using the recursion  $\logit \alpha_{T+k+1}^{(i)}=\omega +\beta \logit \alpha_{T+k}^{(i)}+\tau s_{t+k}^{(i)}$.
\end{enumerate}

 Similarly, point forecasts  $h$ steps ahead  can be obtained  approximating the conditional expectation $E(y_{T+h}|y_{T},\alpha_t)$ with the sample average $B^{-1}\sum_{i=1}^By_{T+h}^{(i)}$. Alternatively, the sample  median of  $y_{T+h}^{(i)}$, $i=1,\dots, B$, can be considered to obtain integer forecasts that are coherent with the discreteness of the data, see \cite{Freeland2004427}.


\section{Some statistical properties}
\label{section2}

In this section, we discuss the reliability of  ML estimation. In particular, we show that the static parameter vector as well as the conditional pmf can be consistently estimated. We focus our asymptotic results on the case of model misspecification. As mentioned before, model misspecification is particularly relevant for models with score-driven parameters. This because  they should be interpreted as filters to approximate a more complex and unknown true DGP  \citep{blasq2015}. The consistency of the ML estimator is therefore obtained with respect to a pseudo-true parameter that has the interpretation of minimizing an average KL divergence between the GAS-INAR model  and an unknown true DGP. Consistency arguments with respect to pseudo-true parameters go back to \cite{white1982}.   In the following, we shall  only assume that the observed data are generated by a stationary and ergodic count process without  imposing a specific DGP.

\subsection{Stability of the filter}

A key ingredient to ensure the reliability of the ML estimator for observation-driven models is the stability of the filtered time-varying parameter. The stability of the filter is typically referred in the literature as the invertibility of the model, see \cite{SM2006} and \cite{Win2013}. As a first step, we derive conditions to ensure that the filtered parameter in (\ref{filter}) converges to a unique stationary  sequence irrespective of the initialization $\hat \alpha_0(\theta)$. This result is particularly important as it implies that  the initialization is irrelevant asymptotically and  provides the basis to ensure the consistency of the ML estimator.

First, we  impose some regularity conditions on the pmf  of the error term $p_e(x,\xi)$. 
\begin{assumption}
\label{assumption1}
The function $\xi \mapsto p_e(x,\xi)$ is continuous in $\Xi$ for any $x \in \mathbb{N}$ and $p_e(x,\xi)>0$ for any $(x,\xi) \in \mathbb{N}\times \Xi$.
\end{assumption}
 Assumption \ref{assumption1} requires the pmf  $p_e(x,\xi)$ to have full support in $\mathbb{N}$ and 
to be continuous  with respect to $\xi$. These conditions are satisfied for most parametric pmf such as the Poisson, the zero inflated Poisson and the Negative Binomial. However, it is worth mentioning that  distributions with limited support such as the Binomial are ruled out by this assumption. 

The next result ensures the stability of the filtered parameter $\{\hat \alpha_t(\theta) \}_{t\in \mathbb{N}}$ specified in (\ref{filter}). In particular, it shows the exponential almost sure  (e.a.s.)~uniform convergence of the functional sequence $\{\hat \alpha_t\}_{t \in \mathbb{N}}$ to a unique stationary and ergodic functional sequence $\{\tilde \alpha_t\}_{t \in \mathbb{Z}}$.  The convergence is  considered with respect to the uniform norm $\|\cdot\|_\Theta$,  where $\|f\|_\Theta=\sup_{\theta \in \Theta }|f(\theta)|$ for any function $f$ that maps from $\Theta$ into $\mathbb{R}$.  We recall that a sequence of non-negative random variables $\{w_t\}_{t\in\mathbb{N}}$ is said to converge e.a.s.~to zero if there exists a constant $\gamma>1$ such that $\gamma^t w_t \xrightarrow{a.s.}$ 0 as $t$ diverges.
\begin{proposition}
\label{proposition1}
Assume that  $\{y_t\}_{t \in \mathbb{Z}}$ is a stationary and ergodic sequence of count random variables such that $E y_t^2 < \infty$. Moreover, let Assumption \ref{assumption1}  be satisfied and let the following condition hold
\begin{eqnarray}
\label{contraction}
E\log\sup_{ \alpha \in (0,1) }|\beta+\tau \dot s_t( \alpha,\xi)|<0, \;\;\forall \; \theta \in \Theta,
\end{eqnarray}
where $\dot s_t(  \alpha,\xi)=\partial s_t ( \alpha, \xi)/\partial \logit\alpha$.
   Then, the filtered parameter $\{\hat \alpha_t(\theta)\}_{t \in \mathbb{N}}$ defined in (\ref{filter}) converges e.a.s.~and uniformly in $\Theta$ to a unique stationary and ergodic sequence $\{ \tilde \alpha_t (\theta)\}_{t \in \mathbb{Z}}$, i.e. 
$$\|\logit \hat \alpha_t  - \logit \tilde \alpha_t \|_\Theta \xrightarrow{\text{e.a.s.}} 0\;\; \text{as} \;\; t\rightarrow \infty,$$
for any initialization $\hat \alpha_0$ of the filter.
\end{proposition}
  The proof is given in the appendix. Proposition \ref{proposition1} does not require correct specification of the model. The observed data  can be  generated by any stationary and ergodic count process.

The contraction condition in (\ref{contraction}) can be checked empirically using the observed data. It  is not possible to obtain a closed form expression for (\ref{contraction}) as it  depends on the DGP and on the specification of $p_e(x,\xi)$. However, with the next proposition, we  show that the parameter region $\Theta$ that satisfies (\ref{contraction})  is not degenerate.

\begin{proposition}
\label{proposition2}
The contraction condition  $(\ref{contraction})$  of Proposition \ref{proposition1} is implied by the following sufficient condition
$$E\log \max(|\beta-\tau y_{t-1}/4|,|\beta+\tau m_t^2|)<0, \;\;\forall \; \theta \in \Theta,$$
where $m_t=\min\{y_{t-1}, y_t\}$.
\end{proposition}
 Proposition \ref{proposition2}  guarantees that the parameter region $\Theta$ is not degenerate as for small enough $|\beta|$ and $|\tau|$ the inequality is always satisfied.


\subsection{Consistency of  ML estimation}

We assume the observed data to be a realized path from an unknown  DGP $\{y_t\}_{t \in \mathbb{Z}}$. Furthermore, we denote with     $p^o(x|y^{t-1})$, $x\in \mathbb{N}$, the true pmf of  $y_t$ conditionally on the past observations $y^{t-1}=\{y_{t-1}, y_{t-2},\dots\}$. The  KL divergence between the true conditional pmf $p^o(x|y^{t-1})$ and the postulated  pmf $p(x|\tilde \alpha_t(\theta),y_{t-1},\xi)$ is given by
$$KL_t(\theta)=\sum_{x=0}^\infty \log \left( \frac{p^o(x|y^{t-1})}{p(x|\tilde \alpha_t(\theta),y_{t-1},\xi)} \right) p^o(x|y^{t-1}).$$
Note that conditional  KL divergence $KL_t(\theta)$ depends on $t$ as it is a function of the past observations $y^{t-1}$.
We are now ready to  formally define  the pseudo-true parameter $\theta^*$.
\begin{definition}
\label{definition1}
The pseudo-true parameter $\theta^*$ is the  minimizer of the average  KL divergence $KL(\theta)= {E}KL_t(\theta)$  in the parameter set $\Theta$.
\end{definition}
 We also denote with $\alpha_t^*=\tilde \alpha_t(\theta^*)$ the pseudo-true dynamic survival probability and with $p_t^*(x)=p(x| \alpha_t^*,y_{t-1},\xi^*)$, $x\in \mathbb{N}$, the pseudo-true conditional pmf.  In the following, we also prove the consistency of the  plug-in estimators $\hat \alpha_t(\hat \theta_T)$  and $\hat p_t(x,\hat \theta_T)=p(x|y_{t-1},\hat \alpha_t(\hat \theta_T),\hat \xi_T)$ for the time-varying survival probability  and conditional pmf  respectively. This is of practical interest  as typically the main objective of INAR models is not the interpretation of the static parameter estimates but approximating the true pmf for  forecasting purposes. 

We start considering the following assumption, which imposes some moment conditions and the contraction condition  of Proposition \ref{proposition1}.
\begin{assumption}
\label{assumption1.5}
The following moment conditions hold true $Ey^2_t<\infty$, $E|\log p^o(y_t|y^{t-1})|<\infty$ and  $E\sup_{\theta\in \Theta}|\log p_e(y_t,\xi)|<\infty$. Furthermore, the  contraction condition in (\ref{contraction}) is satisfied.
\end{assumption}
 Assumption \ref{assumption1.5} is needed to ensure the uniform a.s.~convergence of the  likelihood function $\hat L_T(\theta)$ to a well defined deterministic function $L(\theta)=E l_0(\theta)$, where $l_t(\theta)=\log p(y_t|\tilde \alpha_t(\theta),y_{t-1},\xi)$ denotes the $t$-th contribution to the likelihood function when the limit filter $\tilde \alpha_t(\theta)$ is considered. Furthermore, the integrability condition on the  unknown true pmf $E|\log p^o(y_t|y^{t-1})|<\infty$ is required to ensure that the average KL divergence exists and thus the maximizer of  $L(\theta)$ corresponds to the pseudo-true parameter $\theta^*$.

Note also  that the uniform moment condition $E\sup_{\theta\in \Theta}|\log p_e(y_t,\xi)|<\infty$ is needed  only because we are  considering a general class of pmf for the error term. For most  pmf, this condition is always satisfied. For instance, it holds true   immediately   as long as $Ey_t^2<\infty$ if  $ p_e(x,\xi)$ is a Poisson or a  Negative Binomial pmf. 

Finally, we impose the following identifiability condition.
\begin{assumption}
\label{assumption2}
The function $L(\theta)=E l_0(\theta)$ has a unique maximizer in the set $\Theta$.
\end{assumption}
Assumption \ref{assumption2}  ensures the uniqueness of the pseudo-true parameter $\theta^*$. In general, if this assumption is not satisfied, we   obtain that the limit points of the ML estimator belong to the set of points that minimize the average KL divergence  $KL(\theta)$. 

We are now ready to deliver the strong consistency of the ML estimator  with respect to the pseudo-true parameter $\theta^*$.
\begin{theorem}
\label{theorem1}
Let the observed data $\{y_t\}_{t=1}^T$  be generated by a stationary and ergodic count process $\{y_t\}_{t\in \mathbb{Z}}$ and let the assumptions \ref{assumption1}-\ref{assumption2}  be  satisfied. Then the ML estimator defined in (\ref{MLE}) is strongly consistent with respect to the pseudo-true parameter $\theta^*$, i.e.
$$\hat{\theta}_T \xrightarrow{\text{a.s.}}\theta^*, \;\;\;\; T\xrightarrow{}\infty.$$
\end{theorem}

As special case  of Theorem \ref{theorem1}, we could  also obtain the strong consistency of the ML estimator when the model is correctly specified.
\begin{remark}
\label{corollary1}
If we assume that  the observed data  $\{y_t\}_{t=1}^T$  are generated by a stationary and ergodic  process $\{y_t\}_{t\in \mathbb{Z}}$ that satisfies the model's equations (\ref{m1}) and (\ref{m2}) for $\theta=\theta_0$, $\theta_0\in \Theta$. It can be easily shown that under Assumptions \ref{assumption1}-\ref{assumption2} the ML estimator is strongly consistent.
\end{remark}
In the next section, the finite sample properties of the ML estimator under correct specification are investigated through a simulation study.

 We now turn our attention to the study of the consistency of  the plug-in estimators $\hat \alpha_t(\hat \theta_T)$  and $\hat p_t(x,\hat \theta_T)$.  Note that the consistency of these estimators do not follow trivially from  the consistency of $\hat\theta_T$. This because  these plug-in estimators are random functions of $\hat\theta_T$  that change at different time $t$ without converging. Therefore, it is not possible to trivially  apply a continuous mapping theorem and immediately obtain the desired consistency. The results we obtain require that both $t$ and the sample size $T$ go to infinity. This because $T\rightarrow\infty$ is needed for the consistency of the ML estimator and $t\rightarrow\infty$ is needed to make the effect of the initialization of the filter to vanish.

 The next result shows that the plug-in estimator $\hat \alpha_t(\hat \theta_T)$ is  strongly consistent with respect to the pseudo-true survival probability $\alpha_t^*$.
\begin{lemma}
\label{lemma1.0}
Let the conditions of Theorem \ref{theorem1} hold.  Then, the  plug-in estimator $\logit \hat \alpha_t(\hat \theta_T)$ is strongly consistent, i.e.
$$\left|\logit\hat \alpha_t(\hat \theta_T)-\logit\alpha_t^*\right| \xrightarrow{\text{a.s.}}0,\;\; \;\;\;\; t\xrightarrow{}\infty,\; T\xrightarrow{}\infty.$$
\end{lemma}

In order to obtain the consistency of the   plug-in estimator $\hat p_t(x,\hat \theta_T)$,  we need the following additional regularity condition on the pmf of the error term.
\begin{assumption}
\label{assumption2.5}
The function $\xi \mapsto p_e(x, \xi)$  is continuously differentiable in $\Xi$ for any $x\in \mathbb{N}$.
\end{assumption}
Assumption \ref{assumption2.5} is a standard regularity condition that  is satisfied for most popular pmf such as the Poisson and the Negative Binomial. The next result delivers the consistency of the conditional pmf estimator. In this case, we are only able to ensure consistency and not strong consistency.
\begin{theorem}
\label{theorem2}
Let the observed data $\{y_t\}_{t=1}^T$  be generated by a stationary and ergodic count process $\{y_t\}_{t\in \mathbb{Z}}$ and let the assumptions \ref{assumption1}-\ref{assumption2.5}  be  satisfied. Then the conditional pmf plug-in estimator $\hat p_t(x,\hat \theta_T)$ is consistent, i.e.
$$|\hat p_t(x,\hat \theta_T)- p^*_t(x)| \xrightarrow{\text{pr.}}0,\;\; \;\;\;t\xrightarrow{}\infty,\; T\xrightarrow{}\infty,$$
for any $x \in \mathbb{N}$.
\end{theorem}


\section{Monte Carlo experiment}
\label{section3}

\subsection{Finite sample behavior of the ML estimator}
\label{sim:MLE}

We first perform a Monte Carlo simulation experiment to test the reliability of the ML estimator in finite samples. We consider the dynamic INAR model specified in (\ref{m1}) and (\ref{m2}) with a Poisson error distribution having mean $\mu$. The experiment consists on generating $1000$ time series of size $T$ from the GAS-INAR model and estimating the parameter vector $\theta=(\omega, \beta, \tau, \mu)^T$  by maximum likelihood.  Different parameter values $\theta$ and different sample sizes $T$ are considered. The simulation results are collected in Table \ref{tab1}. In particular, Table \ref{tab1} reports the mean, the bias, the Standard Deviation (SD) and the square root of the Mean Squared Error (MSE) of the ML estimator obtained from the $1000$ Monte Carlo replications.

\begin{table}[h!]
\centering
\resizebox{0.90\columnwidth}{!}{
\begin{tabular}{llccccccccc}
 \cline{3-11}
\vspace{-0.4cm} \\
 & &  $\omega$ & $\beta$ & $\tau$ & $\mu$ &  & $\omega$ & $\beta$ & $\tau$ & $\mu$ \\
  \hline

\vspace{-0.3cm} \\

\multicolumn{2}{c}{\textbf{True Value}} & \textbf{-0.50} & \textbf{0.90} & \textbf{0.15} & \textbf{6.00} &   & \textbf{-0.50} & \textbf{0.95} & \textbf{0.15} & \textbf{6.00}

\vspace{0.1cm} \\

\multirow{ 4}{*}{$T=250$} & Mean & -0.505 & 0.825 & 0.161 & 5.985 &  & -0.496 & 0.896 & 0.159 & 5.996 \\ 
 & Bias & -0.005 & -0.075 & 0.011 &-0.015 &  & 0.004 & -0.054 & 0.009 & -0.004 \\ 
 & SD & 0.326 & 0.175 & 0.100 & 0.588 &  & 0.411 & 0.117 & 0.097 & 0.570  \\ 
   & $\sqrt{\text{MSE}}$  & 0.326 & 0.190 & 0.101 & 0.588 &  & 0.411 & 0.129 & 0.097 & 0.570 \vspace{0.15cm} \\ 

\multirow{ 4}{*}{$T=500$}   &Mean & -0.496 & 0.868 & 0.153 & 5.986 &  & -0.503 & 0.927 & 0.154 & 5.997 \\ 
 & Bias &0.004 & -0.032 & 0.003 & -0.014 &  & -0.003 & -0.023 & 0.004 & -0.003 \\ 
 & SD & 0.213 & 0.093 & 0.062 & 0.407 &  & 0.246 & 0.053 & 0.053 & 0.393\\ 
   & $\sqrt{\text{MSE}}$ & 0.213 & 0.098 & 0.062 & 0.407 &  & 0.246 & 0.058 & 0.053 & 0.392 \vspace{0.15cm} \\ 

\multirow{ 4}{*}{$T=1000$}  &Mean & -0.494 & 0.885 & 0.151 & 5.987 &  & -0.499 & 0.939 & 0.150 & 5.992\\ 
 & Bias & -0.006 & -0.015 & 0.001 & -0.013 &  & -0.001 & -0.011 & 0.000 & -0.008 \\ 
 & SD & 0.152 & 0.050 & 0.042 & 0.295 &  & 0.171 & 0.034 & 0.035 & 0.279 \\ 
   & $\sqrt{\text{MSE}}$ & 0.152 & 0.052 & 0.042 & 0.295 &  & 0.171 & 0.036 & 0.035 & 0.279
\vspace{0.15cm} \\ 
   \hline

\vspace{-0.3cm} \\
\multicolumn{2}{c}{\textbf{True Value}} & \textbf{-0.50} & \textbf{0.90} & \textbf{0.30} & \textbf{6.00} &   & \textbf{-0.50} & \textbf{0.95} & \textbf{0.30} & \textbf{6.00}

\vspace{0.10cm} \\ 

\multirow{ 4}{*}{$T=250$} & Mean & -0.481 & 0.862 & 0.304 & 5.943 &  & -0.502 & 0.916 & 0.302 & 5.945 \\ 
 & Bias & 0.019 & -0.038 & 0.004 & -0.057 &  & -0.002 & -0.034 & 0.002 & -0.055 \\ 
 & SD & 0.361 & 0.095 & 0.101 & 0.512 &  & 0.501 & 0.066 & 0.097 & 0.473  \\ 
   & $\sqrt{\text{MSE}}$ & 0.361 & 0.103 & 0.101 & 0.514 &  & 0.500 & 0.075 & 0.097 & 0.476  \vspace{0.15cm} \\ 

\multirow{ 4}{*}{$T=500$}   &Mean & -0.495 & 0.883 & 0.297 & 5.971 &  & -0.492 & 0.935 & 0.298 & 5.971 \\ 
    & Bias & 0.005 & -0.017 & -0.003 & -0.029 &  & 0.008 & -0.015 & -0.002 & -0.055 \\ 
 & SD & 0.221 & 0.044 & 0.057 & 0.338 &  & 0.361 & 0.030 & 0.052 & 0.310 \\ 
   & $\sqrt{\text{MSE}}$ & 0.221 & 0.048 & 0.057 & 0.339 &  & 0.361 & 0.033 & 0.052 & 0.311 \vspace{0.15cm} \\ 

 \multirow{ 4}{*}{$T=1000$} &Mean & -0.490 & 0.891 & 0.299 & 5.978 &  & -0.502 & 0.943 & 0.298 & 5.981  \\ 
   & Bias & 0.010 & -0.019 & -0.001 & -0.022 &  & -0.002 & -0.007 & 0.002 & -0.019 \\ 
 & SD & 0.156 & 0.029 & 0.040 & 0.242 &  & 0.233 & 0.019 & 0.035 & 0.219 \\ 
   & $\sqrt{\text{MSE}}$ & 0.156 & 0.031 & 0.040 & 0.243 &  & 0.233 & 0.020 & 0.035 & 0.220  \\ 

   \hline

\end{tabular}
}

\caption{\textit{Summary statistics of the sample ML estimator distribution  for different parameter values $\theta$ and different sample sizes $T$. The statistics in the table are obtained from 1000 Monte Carlo replications.}}
\label{tab1}
\end{table}


The simulation results in Table \ref{tab1} further suggest that the parameter vector $\theta$ can be consistently estimated by maximum likelihood. This can be elicited from the fact that the MSE of the estimator is decreasing as the sample size $T$ increases. We also note that the estimator of the parameter $\beta$ tends to be negatively biased in finite samples. In all the cases considered, the parameter $\beta$ is underestimated on average. The magnitude of the  bias seems also to be relevant as, especially for $T=250$, the square root of the MSE is considerably  larger then the SD. Therefore, this indicates that the bias contribution  to the MSE is not negligible compared to the variance contribution.  The negative bias for $\beta$  is not surprising as the values of $\beta$ considered in the simulations are close to 1 and similar results on the bias  are well known  for ML estimation of linear autoregressive models. As concerns the other parameters, the results suggest that the bias can be considered negligible as the SD is almost equal to the square root of the MSE in all the scenario considered.  

\subsection{Filtering under misspecification}
Score-driven updates  for  time-varying parameters have been shown to be optimal in a  misspecified framework where the aim is to reduce the KL divergence between the postulated model and the true unknown  DGP, see \cite{blasq2015}. This section illustrates the flexibility of the proposed GAS-INAR specification through a simulation study. In this experiment, we consider  different  DGPs of the form
$$y_t= \alpha^o_t \circ y_{t-1} + \varepsilon_t, \;\;\; \varepsilon_t\sim \mathcal{P}(5),$$
where $\mathcal{P}(5)$ denotes a Poisson distribution  with mean equal to 5.  The DGPs differ on the basis of the  specification of the sequence $\{\alpha_t^o\}_{t \in \mathbb{Z}}$. The following four dynamics are considered.
\begin{enumerate}
\item Fast sine:   $\alpha_t^o =0.5 + 0.25\sin(\pi t/100) $.
\item Slow sine:  $\alpha_t^o =0.5 + 0.25\sin( \pi t/250) $.
\item Fast steps: $\alpha_t^o=0.25I_{[-1,0]}\left(\sin( \pi t/100)\right)+0.75I_{(0,1]}\left(\sin( \pi t/100)\right) $.
\item Slow steps:  $\alpha_t^o=0.25I_{[-1,0]}\left(\sin( \pi t/250)\right)+0.75I_{(0,1]}\left(\sin( \pi t/250)\right) $.
\end{enumerate}
where $I_A(x)=1$ if $x\in A$ and $I_A(x)=0$ otherwise. The DGPs are thus Poisson INAR models where the coefficient $\alpha_t^o$ is allowed to change in different ways. The red lines in Figure \ref{filtersim} show the path of $\alpha_t^o$, $t=1,\dots, 500$, for the four different DGPs. As we can see, the fast sine and the slow sine specifications allow the coefficient to change smoothly over time, whereas, the fast step and slow step specifications exhibit abrupt changes over time.

\begin{table}[ht]
\centering
\begin{tabular}{lcccc}
   & \multicolumn{4}{c}{\textbf{Square root MSE}} \\ 
  \vspace{-0.5cm}\\
    \cline{2-5}
  \vspace{-0.4cm}\\
 & Fast sine  & Slow sine & Fast steps & Slow steps \\ 
  \hline
  \vspace{-0.35cm}\\
 INAR  & 0.242 &  0.257 & 0.322 & 0.356 \\ 
  \vspace{-0.4cm}\\
 rc-INAR &  0.112 & 0.111 & 0.145 & 0.132  \\ 
  \vspace{-0.4cm}\\
 GAS-INAR  & \textbf{0.077} &   \textbf{0.060} &  \textbf{0.101} &  \textbf{0.072}  \\ 
  \vspace{-0.4cm}\\
\hline
  \vspace{-0.2cm}\\
   & \multicolumn{4}{c}{\textbf{KL divergence}} \\ 
  \vspace{-0.5cm}\\
    \cline{2-5}
  \vspace{-0.4cm}\\
 & Fast sine  & Slow sine & Fast steps & Slow steps \\ 
  \hline
  \vspace{-0.35cm}\\
 INAR  & 0.238 & 0.253 & 0.412 & 0.442 \\ 
  \vspace{-0.4cm}\\
 rc-INAR & 0.117 & 0.114 &0.212 & 0.185 \\ 
  \vspace{-0.4cm}\\
 GAS-INAR  & \textbf{0.053} &   \textbf{0.029}&  \textbf{0.128} & \textbf{0.057} \\ 
  \vspace{-0.4cm}\\
  \hline
\end{tabular}
\caption{\textit{Average MSE and KL divergence between the true DGP and the different models. } }
\label{sim:filter}
\end{table}
\begin{figure}[h!]
\center
\includegraphics[scale=0.60]{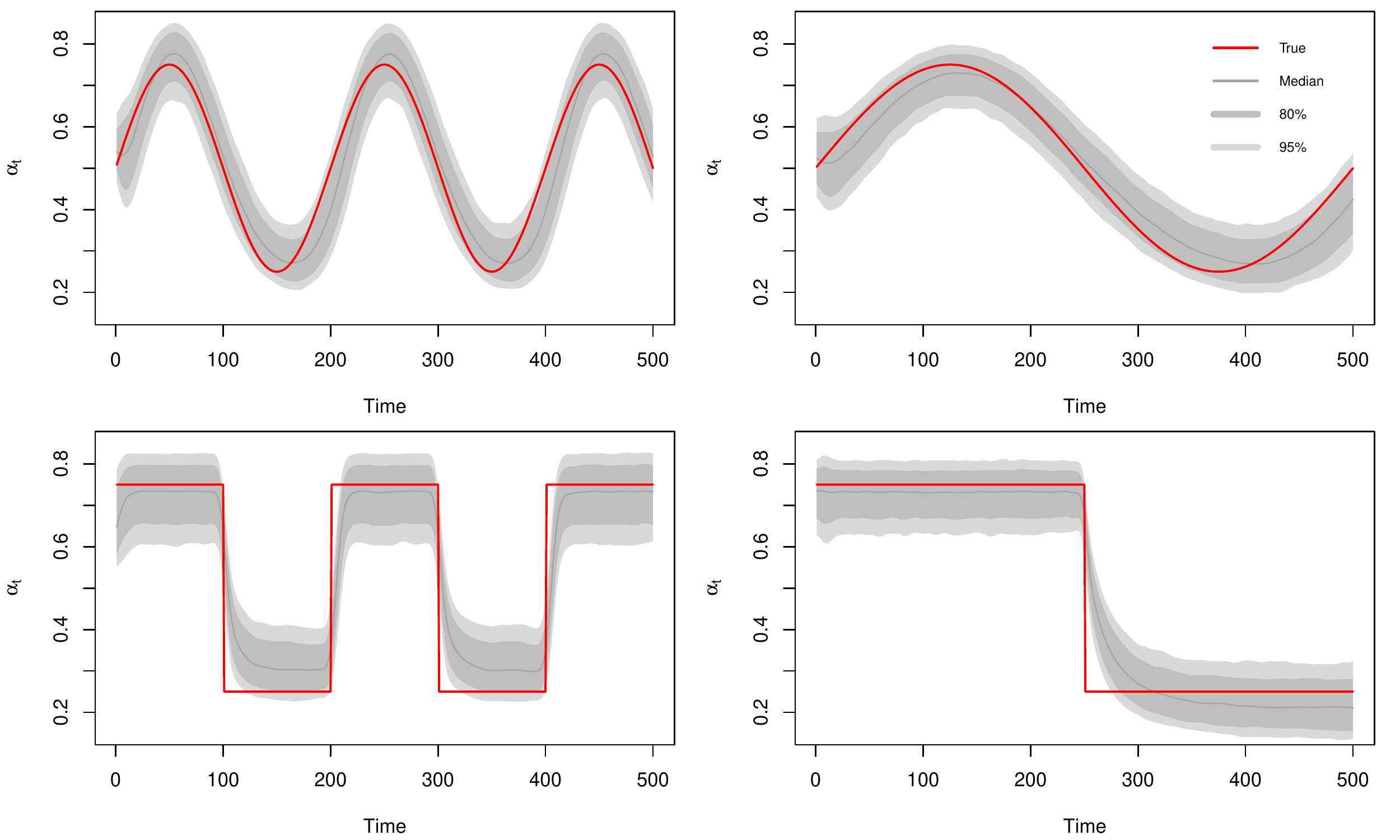}

\caption{\textit{The red line denotes the true path $\alpha_t^o$. The gray area represents confidence bounds of the filtered path of $\alpha_t$ for the GAS-INAR model. The first plot is for the fast sine configuration, the second is for the slow sine, the third is for the fast steps and the last is for the slow steps specification. }  }
\label{filtersim}
\end{figure}

The simulation  experiment consists on  generating $1000$  Monte Carlo time series draws  of size $T=500$ from the different DGPs. For each draw, the following models are estimated: a Poisson INAR model with static coefficient, the  GAS-INAR model  with Poisson error terms and a Poisson INAR model with dynamic coefficient as considered  in \cite{zheng2008first}. For the latter model the dynamic survival probability is given by $\logit \alpha_t=\omega+\tau y_{t-1}$, where $\omega$ and $\tau$ are parameters to be estimated. The model of \cite{zheng2008first} is denoted as rc-INAR.  The performances of the models is measured in terms of approximation of the true conditional pmf and the true survival probability $\alpha_t^o$. As concerns pmf approximation, we compute the KL divergence between the true pmf and the estimated one. Whereas, as concerns $\alpha_t^o$, we consider the MSE between $\alpha_t^o$ and the estimated survival
 probability. Table \ref{sim:filter} reports the results of the simulation experiment. As we can see, the GAS-INAR model has the pest performance in terms of both KL divergence and MSE. This is true for all the DGPs considered. We also note that the better performance of the GAS-INAR model is relevant in relative terms. In particular, the KL divergence and MSE from the GAS-INAR model are about half of those from the rc-INAR model and about one third of those from the INAR model.    These results show the flexibility of the GAS-INAR model and its ability to approximate complex DGPs.

Figure \ref{filtersim} further illustrates the ability of the GAS-INAR specification to capture the dynamic behavior of the true $\alpha_t^o$ in the different settings considered. The gray areas in the plots represent  $95\%$ variability bounds for the estimated paths of $\alpha_t$ and the red lines denote the true paths $\alpha_t^o$. As we can see, in the fast sine and slow sine configurations, the true path $\alpha_t^o$ is always inside the $95\%$ confidence bounds. This shows the ability of the GAS-INAR model to capture smooth changes in $\alpha_t^o$. On the other hand, in the fast steps and slow steps configurations, the true  $\alpha_t^o$  is not  inside the confidence bounds right after the sudden changes in the level of $\alpha_t^o$. This is natural as the filtered path requires some time periods before adapting to the break in the level of $\alpha_t^o$. However, also in this situation, we can see how the estimated paths from the GAS-INAR model are able to approximate reasonably well the true $\alpha_t^o$.

\section{Application to crime data}
\label{section4}

\subsection{In-sample results}
We present an empirical illustration of the proposed methodology to the monthly  number of offensive conduct reports   in the city of Blacktown, Australia, from January 1995 to December 2014. The time series is    from the New South Wales  dataset of police reports and it is available at \url{http://data.gov.au/}.
\begin{figure}[h!]
\center
\includegraphics[scale=0.55]{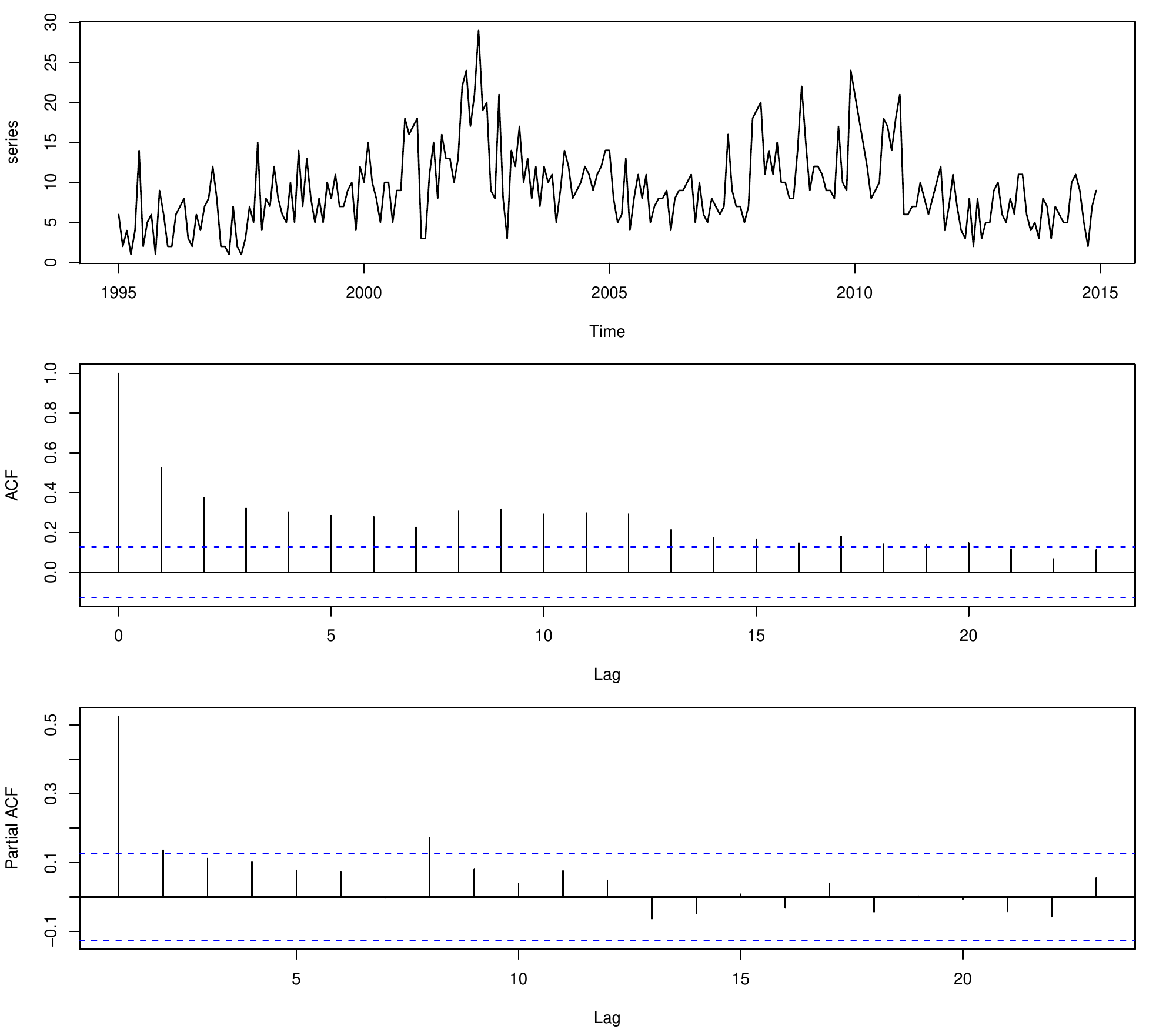}
\vspace{-0.2cm}
\caption{\textit{The first plot shows the monthly number of offensive conduct reports in Blacktown from January 1995 to December 2014. The second and third plots represent the sample autocorrelation functions of the series.}}
\label{f1}
\end{figure}
 Figure \ref{f1} shows the plot of the series. As we can see, there are  two time periods with a particular high level of criminal activities. The first is around 2002 and the second is around 2010. During these periods  we  expect the estimated survival probability $\alpha_t$ to be higher as they can be seen as periods of high persistence.
As discussed in \cite{jin1991}, INAR(p) models have the same autocorrelation structure of continuous-valued AR(p) models. The sample  autocorrelation functions  in Figure \ref{f1}  suggest that a first-order INAR model should be appropriate for this dataset.  We consider several model specifications: the  INAR and the GAS-INAR model with  Poisson and Negative Binomial  error distribution.   The  sample mean of the data  is $9.3$ and the sample variance  is $24.3$.  This is an indication that   there is overdispersion in the data and thus a Negative Binomial distribution for the error term may be more suited.   The different  specifications employed are summarized in Table \ref{spec}.

\begin{table}[h!]
\centering
\resizebox{0.95\columnwidth}{!}{
\begin{tabular}{ll}
  \cline{2-2}
\vspace{-0.4cm}\\
  & \multicolumn{1}{c}{Model description}   \\ 
  \hline
\vspace{-0.35cm}\\
GAS-NBINAR & Model in (\ref{m1}) and  (\ref{m2}) with Negative Binomial error of mean $\mu$ and variance $\sigma^2$.  \\ 
\vspace{-0.4cm}\\
NBINAR & Model in (\ref{inar}) with Negative Binomial error of mean $\mu$ and variance $\sigma^2$.  \\ 
\vspace{-0.4cm}\\
GAS-PoINAR & Model in (\ref{m1}) and  (\ref{m2}) with Poisson error of mean $\mu$.   \\ 
\vspace{-0.4cm}\\
PoINAR &  Model in (\ref{inar}) with Poisson error of mean $\mu$.   \\ 
   \hline
\end{tabular}
}

\caption{\textit{The table describes the specification of each model.}}
\label{spec}
\end{table}

\begin{table}[ht]
\centering
\begin{tabular}{lcccccccc}
 \cline{2-9}
\vspace{-0.4cm} \\
 & $\omega$ & $\beta$ & $\tau$ & $\mu $ & $\sigma^2 $  & log-lik & pvalue & AIC \\ 
 \hline \vspace{-0.35cm} \\
GAS-NBINAR  & -0.907 & 0.965 & 0.135 & 6.083 & 14.155& -662.91 &0.002&  1335.82 \vspace{-0.15cm} \\ 
  & \footnotesize{(0.338)} & \footnotesize{(0.027)} & \footnotesize{(0.055)} & \footnotesize{(0.481)} &\footnotesize{(1.853)}&  &&  \vspace{0.1cm}\\ 
 NBINAR  & -0.401 & - & - & 5.586 & 15.265 & -669.03 &-& 1344.07  \vspace{-0.15cm} \\  
  & \footnotesize{(0.176)} &  &  & \footnotesize{(0.456)}  &\footnotesize{(2.125)}&  &&  \vspace{0.1cm}\\ 
GAS-PoINAR  & -1.258 & 0.967 & 0.141 & 6.539& - & -695.04 &0.000& 1398.24 \vspace{-0.15cm} \\ 
  & \footnotesize{(0.294)} & \footnotesize{(0.019)} & \footnotesize{(0.033)} & \footnotesize{(0.313)} && & &  \vspace{0.1cm}\\ 
 PoINAR  & -0.613& - & - &  6.046 & - & -714.58 &-& 1433.21  \vspace{-0.15cm} \\  
  & \footnotesize{(0.140)} &  &  & \footnotesize{(0.323)} && & &  \\ 
   \hline
\end{tabular}

\caption{\textit{ML estimate of the models in Table \ref{spec}. The last three columns contain respectively the log-likelihood, the pvalue of the likelihood ratio test between the GAS-INAR models and their static INAR counterparts  and the AIC.}}
\label{tab2}
\end{table}

 The ML estimation  results    are collected in Table \ref{tab2}.  We consider the likelihood ratio test to check the significance of the dynamic coefficient $\alpha_t$. Given its meaningful interpretation in a misspecified framework, we  also report the Akaike Information Criterion (AIC) as a means of comparison among non-nested models.  The results suggest that the inclusion of the dynamic specification for $\alpha_t$  plays a relevant role as confirmed  by  the likelihood ratio test and the AIC.  The likelihood ratio test shows that the dynamic coefficient is highly significant for both the Poisson and the Negative Binomial specifications. Overall the model with the smallest AIC is the GAS-NBINAR model.  Furthermore, for both the  Negative Binomial models, the estimated variance of the error term is more than double the estimated mean.  We can thus say that the Negative Binomial distribution seems to provide a better fitting than the Poisson. This result   is  also coherent with the overdispersion observed in the data. We  can conclude that  the results indicate a better in-sample performance for the GAS-INAR model.

 From Table \ref{tab2}, we also note that the time-varying parameter $\alpha_t$ is highly persistent as the estimated $\beta$ is close to 1. The estimated path of $\alpha_t$  together with $80\%$ and $95\%$ confidence bounds  is plotted in Figure \ref{at}. As expected, the  survival probability is particularly high around 2002 and around 2010. This reflects the high level of criminal activities  that can be interpreted as an higher survival probability of past elements. The plot in Figure \ref{at} also highlights that there is a relevant difference in considering a static $\alpha$ instead of a dynamic $\alpha_t$. This can be noted from the fact that the dashed line, which denotes the static parameter estimate of $\alpha$,  lies outside the 95\% confidence bounds of $\alpha_t$ in some time periods.


\begin{figure}[h!]
\center
\includegraphics[scale=0.65]{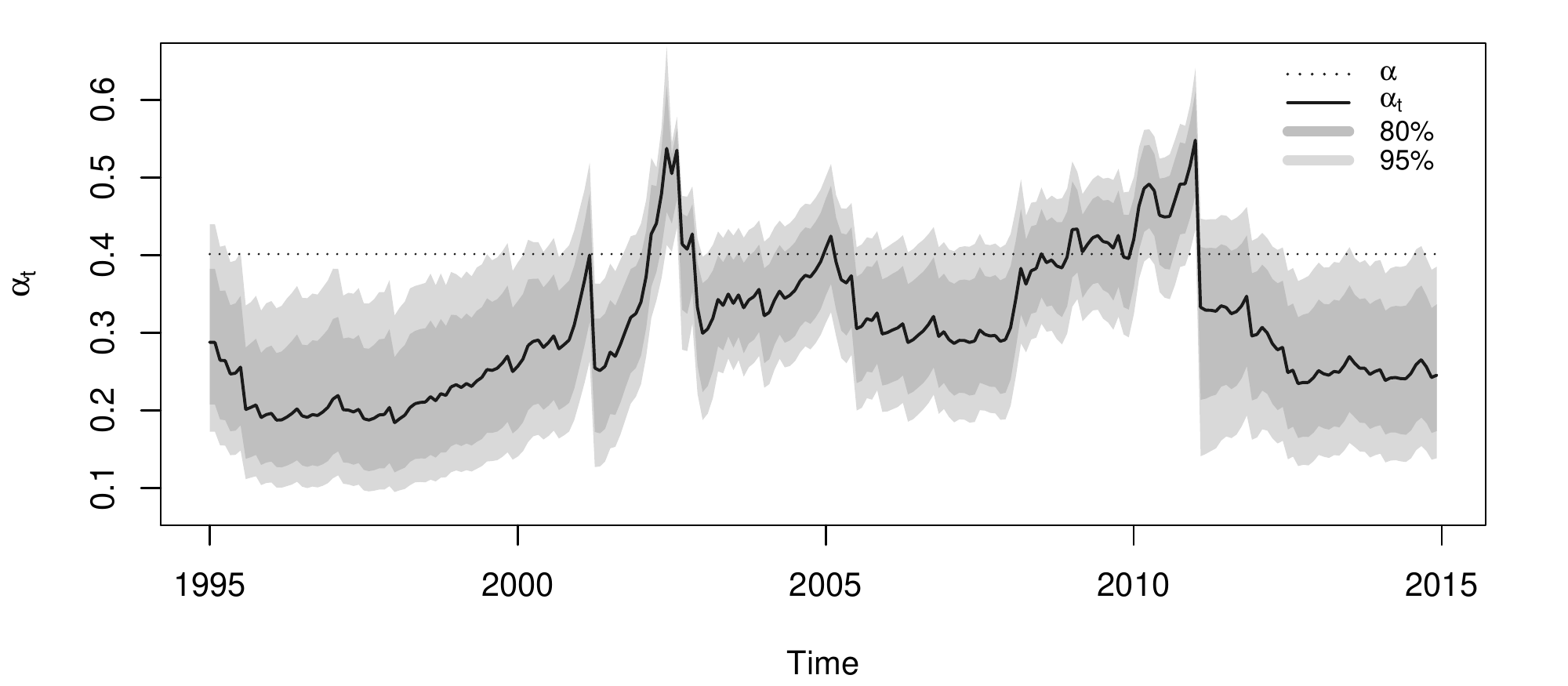}
\vspace{-0.2cm}
\caption{\textit{Estimated  $\alpha_t$ from the GAS-NBINAR model with 80\% and 95\% confidence bounds. The dashed line represents the estimate of $\alpha$  from the NBINAR model. The confidence bounds are obtained simulating from the  distribution of the ML estimator as proposed by \cite{cb2016}. }}
\label{at}
\end{figure}


\subsection{Forecasting results}

Finally we perform a pseudo out-of-sample experiment to compare the forecasting performances of the models. The full sample size of the series is  240 observations. We split it into  two subsamples: the first 140 observations are considered as a training sample and the last 100  observations as a forecasting evaluation sample. The training sample is then expanded recursively. We evaluate the forecast performance of the models in terms of both point forecast and pmf forecast. The point forecast accuracy is evaluated by the forecast MSE, i.e.~$100^{-1}\sum_{i=1}^{100} (\hat y_{T+i}-y_{T+i})^2 $. Whereas, the pmf forecast accuracy is evaluated by the log score criterion, i.e. $100^{-1}\sum_{i=1}^{100}\log \hat p_{T+i}(y_{T+i})$. The log score criterion provides a means of comparison based on the KL divergence between  the true DGP and the estimated models.  

\begin{table}[h!]
\centering
\begin{tabular}{lcccccc}

&  \multicolumn{6}{c}{{\textbf{Mean squared error}}} \\ 
\vspace{-0.5cm} \\
 \cline{2-7}
\vspace{-0.4cm} \\
 & $h=1$ & $h=2$ & $h=3$ & $h=4$ & $h=5$ & $h=6$ \\ 
   \hline \vspace{-0.35cm} \\
GAS-NBINAR & \textbf{15.77} & \textbf{20.15} & \textbf{20.56} & \textbf{21.51} & \textbf{21.36} & \textbf{21.23} \vspace{0.1cm}\\  
NBINAR & 16.51 & 21.47 & 22.61 & 23.70 & 23.85 & 23.72  \vspace{0.1cm}\\ 
GAS-PoINAR & 16.33 & 20.66 & 21.18 & 21.98 & 21.82 & 21.52 \vspace{0.1cm}\\  
PoINAR & 17.00 & 21.82 & 22.86 & 23.79 & 23.91 & 23.78  \vspace{0.1cm}\\ 
\hline 
\vspace{-0.3cm} \\
&  \multicolumn{6}{c}{\textbf{{Log score criterion}}} \\ 
 \cline{2-7}
\vspace{-0.4cm} \\

 & $h=1$ & $h=2$ & $h=3$ & $h=4$ & $h=5$ & $h=6$ \\ 
   \hline \vspace{-0.35cm} \\
GAS-NBINAR & \textbf{-2.73} & \textbf{-2.82} & \textbf{-2.83} & \textbf{-2.85} & \textbf{-2.85} & \textbf{-2.85} \vspace{0.1cm}\\   
NBINAR & -2.75 & -2.85 & -2.88 & -2.91 & -2.91 & -2.91  \vspace{0.1cm}\\  
GAS-PoINAR & -2.83 & -2.96 & -2.98 & -3.00 & -3.00 & -2.98 \vspace{0.1cm}\\   
PoINAR & -2.88 & -3.08 & -3.12 & -3.18 & -3.19 & -3.18\\ 
   \hline
\end{tabular}
\caption{\textit{Forecast MSE and log score criterion computed in the last 100 observations for different forecast horizons.}}
\label{table:forecasts}
\end{table}

The results  are collected in Table \ref{table:forecasts}. As we can see, the inclusion of the dynamic survival probability $\alpha_t$ provides better forecast performances in the subsample considered. In particular, the GAS-NBINAR model outperforms the NBINAR model   in terms of  both point forecasts and pmf forecasts. The same happens for the GAS-PoINAR compared to the PoINAR. This holds true for all forecast horizons considered.  Furthermore,   the use of the Negative Binomial distribution is particularly relevant to improve the pmf forecasts. In particular, the Negative Binomial models dominate the Poisson models in terms of log-score criterion. This result is quite natural as the the Negative Binomial models take into account the overdispersion in the data.  On the other hand, as concerns the point forecasts, the dynamic parameter $\alpha_t$ seems to play a mayor role in improving the forecast performances. This can be noted as the models with dynamic $\alpha_t$ dominate the models with static $\alpha$ in terms on MSE. The best performing model is the GAS-NBINAR for both criteria and  all forecast horizons. This  suggests that the flexibility introduced by $\alpha_t$ as well as the choice of an appropriate distribution for the error term can be important to better predict future observations. Overall, these out-of-sample results together with the in-sample results show that the GAS-INAR models can be useful in practical applications.

\section{Conclusion}
\label{section5}
In this paper, we have proposed a flexible  INAR model with dynamic survival probability. This model should be interpreted as a filter to approximate  unknown DGPs. Empirical results are promising as illustrated in the empirical experiments  considering   both  simulated data  and   real data. Future  research may include the extension of the first-order dynamic INAR model to a general order $p$. Other work to be done concerns the asymptotic theory of the ML estimator. At the moment, we have only proved the consistency of the estimator. The asymptotic normality requires the study of the first two derivatives of the log likelihood. In this regard, we encountered some difficulties concerning the existence of some moments for the derivative processes.


\appendix

\section{Appendix}

\subsection{Derivatives of the predictive log-likelihood}
\label{appendixA1}

Defining $s_t(\alpha, \xi):=\partial \log p(y_t| \alpha,y_{t-1},\xi)/\partial \logit \alpha$ and $\dot s_t(\alpha, \xi):=\partial s_t(\alpha, \xi)/\partial \logit\alpha$,  by elementary calculus, we obtain that 
\begin{align}
\label{first_der}
s_t(\alpha,\xi)=\left(\sum_{k=0}^{m_t}p_{kt}(\alpha,\xi)\right)^{-1}\left(\sum_{k=0}^{m_t} p_{kt}(\alpha,\xi)(k-y_{t-1}\alpha) \right),
\end{align}
and
\begin{align}
\label{second_der}
\dot s_t(\alpha,\xi)=&\left(\sum_{j=0}^{m_t}\sum_{k=0}^{m_t}p_{kt}(\alpha,\xi)p_{jt}(\alpha,\xi)\right)^{-1}\times \nonumber \\
&\qquad\left(\sum_{j=0}^{m_t}\sum_{k=0}^{m_t}p_{kt}(\alpha,\xi)p_{jt}(\alpha,\xi)\left(k(k-j)-\alpha(1-\alpha)y_{t-1}\right)\right),
\end{align}
where $m_t=\min(y_t,y_{t-1})$ and 
\begin{align*}
p_{kt}(\alpha,\xi)= {y_{t-1}\choose k} \alpha^k (1-\alpha)^{y_{t-1}-k}p_e(y_{t}-k, \xi).
\end{align*}

\subsection{Proofs}

\begin{proof}[Proof of Proposition \ref{proposition1}]

The stability conditions we consider to obtain the convergence result are based on Theorem 3.1 of \cite{Bougerol1993}.   \cite{SM2006} applied  Bougerol's theorem in the space of continuous functions $\mathbb{C}(\Theta,\mathbb{R})$ equipped with the uniform  norm $\|\cdot\|_{\Theta}$. In particular, they provide  stability conditions for functional SRE of the form
\begin{equation}\label{sre_app}
x_{t+1}(\theta)=\phi_t(x_t(\theta),\theta), \; t \in \mathbb{N},
\end{equation}
where $x_0(\theta)\in \mathbb{R}$, the map  $(x,\theta) \mapsto \phi_t(x,\theta)$ from $\mathbb{R}\times \Theta $ into $\mathbb{R}$ is almost surely continuous and   the sequence $\{\phi_t(x,\theta)\}_{t\in \mathbb{Z}}$ is stationary and ergodic for any $(x,\theta)\in \mathbb{R}\times \Theta $.   \cite{Win2013} weakened \cite{SM2006} conditions   replacing a uniform contraction condition with a pointwise condition.  The uniform e.a.s~convergence of a filter satisfying the SRE in (\ref{sre_app}) can be obtained on the basis of Theorem 2 of \cite{Win2013} from the following conditions:
\begin{description}
\item[(a)] There exists an $x \in \mathbb{R}$ such that  $E \log^+\left(\sup_{\theta\in\Theta}|\phi_0(x,\theta)|\right)<\infty$,
\item[(b)] $E \log^+\left(\sup_{\theta\in\Theta}\Lambda_0(\theta)\right)<\infty$,
\item[(c)] $E \log\left(\Lambda_0(\theta)\right)<0$ for any $\theta \in \Theta$,
\end{description}
where the random coefficient $\Lambda_t(\theta)$ is defined as
$$\Lambda_t(\theta)=\sup_{(x_1,x_2)\in \mathbb{R}^2,x_1\neq x_2}\frac{|\phi_t(x_1,\theta)-\phi_t(x_2,\theta)|}{|x_1-x_2|}.$$

In our case, the random function  $\phi_t$ that defines the SRE in (\ref{sre_app}) has the following form
  $$\phi_t(x,\theta)=\omega+\beta x + \tau s_t\left(\logit^{-1} (x),\xi\right).$$
First we note that our SRE satisfies the stationarity and  continuity requirements to apply Wintenberger's results. In particular, we obtain that the a.s.~continuity of $\phi_t(x,\theta)$ follows immediately from the a.s.~continuity of $(x,\theta) \mapsto  s_t\left(\logit^{-1} (x),\xi\right)$, which is implied by Assumption \ref{assumption1}, and the continuity of the Binomial likelihood (see the functional form of  $s_t$ in (\ref{first_der})). Furthermore, the stationarity and ergodicity of  $\{\phi_t\}_{t\in \mathbb{Z}}$ follows from the stationarity and ergodicity of $\{y_t\}_{t\in \mathbb{Z}}$ together with an application of Proposition 4.3 of \cite{krengel1985} as $s_t\left(\logit^{-1} (x),\xi\right)$ is a measurable function of $y_t$ and $y_{t-1}$. In the following,   we will prove the proposition by showing that conditions (a)-(c) are satisfied.

As concerns (a), setting $x=0$ and accounting that $E y_0^2<0$, by an application of Lemma \ref{lemma1},  we obtain that 
\begin{align*}
E \log^+\left(\sup_{\theta\in\Theta}|\phi_0(x,\theta)|\right) &\le \sup_{\theta\in\Theta}|\omega| +\sup_{\theta\in\Theta}|\tau|  E\sup_{\theta\in\Theta}|s_t\left(0.5,\xi\right)|\\
&\le \sup_{\theta\in\Theta}|\omega| +\sup_{\theta\in\Theta}|\tau|  E|y_{t-1}|<\infty.
\end{align*}
Thus (a) is proved.

As concerns (b), by an application of Lemma \ref{lemma1}, we have  that
\begin{align*}
E \log^+\left(\sup_{\theta\in\Theta}\Lambda_0(\theta)\right)&\le E \sup_{\theta\in\Theta}\sup_{x\in\mathbb{R}}|\partial \phi_0(x,\theta)/\partial x|\le \sup_{\theta\in\Theta}|\beta| +\sup_{\theta\in\Theta}|\tau|  E\sup_{\theta\in\Theta}\sup_{\alpha\in(0,1)}|\dot s(\alpha,\xi)|\\
&\le \sup_{\theta\in\Theta}|\beta| +\sup_{\theta\in\Theta}|\tau|  E|y_{t-1}^2|<\infty,
\end{align*}
as $Ey_0^2<\infty$. This shows that (b) holds true.

Finally, as concerns (c), by condition (\ref{contraction}) we obtain for any $\theta \in \Theta$
\begin{align*}
E \log\left(\Lambda_0(\theta)\right)&\le E \sup_{x\in\mathbb{R}}|\partial \phi_0(x,\theta)/\partial x|\le  E\sup_{\alpha\in(0,1)}|\beta +\tau \dot s(\alpha,\xi)|<0.
\end{align*}
This proves (c) and concludes the proof of the proposition.

\end{proof}


\begin{proof}[Proof of Proposition \ref{proposition2}]

The result follows immediately by an application of Lemma \ref{lemma1}, which provides an upper bound for the derivative of the score.

\end{proof}


\begin{proof}[Proof of Theorem \ref{theorem1}]

Assumption \ref{assumption2} ensures that $L(\theta)=E l_t(\theta)$ has a unique maximizer in the compact set $\Theta$, which indeed corresponds to the pseudo-true parameter $\theta^*$ that minimizes $KL(\theta)$ as $E|\log p^o(y_t|y^{t-1})|<\infty$ is satisfied by assumption. In the following, we show that the log-likelihood function $\hat L_T(\theta)$ converges almost surely uniformly in $\Theta$ to $L(\theta)$, namely
\begin{align}
\label{A_uniconv}
\|\hat L_T-L\|_\Theta\xrightarrow{a.s.}0, \;\; T\rightarrow\infty.
\end{align}
Then, given the compactness of $\Theta$ and the identifiability of $\theta^*$, the almost sure convergence $\hat\theta_T\xrightarrow{a.s.}\theta^*$ follows by well known standard arguments due to \cite{wald1949}.

Defining $L_T(\theta)=T^{-1}\sum_{t=1}^Tl_t(\theta) $, with $l_t(\theta)=\log p(y_t|\tilde\alpha_t(\theta),y_{t-1},\xi)$, an application of the triangle inequality yields
\begin{align}
\label{A_triangle}
\|\hat L_T-L\|_\Theta\le \|\hat L_T-L_T\|_\Theta+\| L_T-L\|_\Theta.
\end{align}
Therefore, the uniform convergence in (\ref{A_uniconv}) follows if both terms on the right hand side of the inequality (\ref{A_triangle}) converge almost surely to zero.

First we show that $\|\hat L_T-L_T\|_\Theta\xrightarrow{a.s.}0$. An application of the mean value theorem together with Lemma \ref{lemma1} yields
 \begin{align*}
|\hat l_t(\theta)-l_t(\theta)|&\le \sup_{\alpha\in (0,1)}|s_t(\alpha,\xi)||\logit\hat\alpha_t(\theta)-\logit \tilde \alpha_t(\theta)|\\
&\le y_{t-1}|\logit\hat\alpha_t(\theta)-\logit \tilde \alpha_t(\theta)|
\end{align*}
for any $\theta \in \Theta$ and $t \in \mathbb{N}$. Furthermore,  taking into account that $\|\logit\hat\alpha_t-\logit \tilde \alpha_t\|_\Theta\xrightarrow{e.a.s.}0$ by Proposition \ref{proposition1}  and  that $E|y_{t-1}|<\infty$ holds true by assumption, an application of Lemma 2.1 of \cite{SM2006} yields
 $$\sum_{t=1}^\infty y_{t-1}\|\logit\hat\alpha_t-\logit \tilde \alpha_t\|_\Theta< \infty$$
 almost surely.
As a result, we have that $T^{-1}\sum_{t=1}^T \|\hat l_t-l_t\|_\Theta\xrightarrow{a.s.}0$ and therefore we conclude that the desired result $\|\hat L_T-L_T\|_\Theta\xrightarrow{a.s.}0$ is proved as 
$$\|\hat L_T-L_T\|_\Theta\le T^{-1}\sum_{t=1}^T \|\hat l_t-l_t\|_\Theta.$$

We are now left with showing that $\| L_T-L\|_\Theta\xrightarrow{a.s.}0$.  Note that $\{l_t\}_{t\in \mathbb{N}}$ is a stationary and ergodic sequence of random elements that takes values in the space continuous functions  $\mathbb{C}(\Theta,\mathbb{R})$ equipped with the uniform norm $\|\cdot\|_\Theta$.  Therefore, the desired convergence result follows by an application of the ergodic theorem of \cite{rao1962} provided that  the uniform  integrability condition $E\|l_t\|_\Theta<\infty$ is satisfied. In the following, we show that this condition holds true.
First, note that $ l_t(\theta)\le0$ with probability 1 for any $\theta \in \Theta$ as $p(y_1| \alpha, y_2, \xi)\le 1$ for any $(y_1,y_2,\xi,\alpha)\in\mathbb{N}^2\times \Xi \times (0,1)$. Thus, accounting that $\log(1+\exp(x))\le1+|x|$ for any $x\in\mathbb{R}$, we obtain 
 \begin{align*}
 |l_t(\theta)|&=-l_t(\theta)=-\log \sum_{k=0}^{m_t} p_{kt}(\tilde \alpha_t(\theta),\xi)\le -\log  p_{0t}(\tilde \alpha_t(\theta),\xi)\\
& \le-y_{t-1}\log(1-\tilde\alpha_t(\theta))-\log p_e(y_{t-1},\xi)\\
& \le y_{t-1}\log(1+\exp(\logit\tilde\alpha_t(\theta)))-\log p_e(y_{t-1},\xi)\\
& \le y_{t-1}(1+|\logit\tilde\alpha_t(\theta)|)-\log p_e(y_{t-1},\xi),
\end{align*}
 almost surely for any $\theta \in \Theta$. Finally, an application of the Cauchy-Schwarz inequality yields
$$\|l_t\|\le E y_t + E y_t^2 +\|\logit\tilde\alpha_t\|_\Theta^2+E\sup_{\theta\in\Theta}|\log p_e(y_{t-1},\xi)|<\infty,$$
where $ E y_t^2<\infty$ and $E\sup_{\theta\in\Theta}|\log p_e(y_{t-1},\xi)|<\infty$ are satisfied by assumption and $\|\logit\tilde\alpha_t\|_\Theta^2<\infty$ follows by an application of Lemma \ref{lemma2}.
\end{proof}


\begin{proof}[Proof of Lemma \ref{lemma1.0}]
The proof of this result is an immediate consequence of  Theorem 3 of \cite{Win2013}. We simply  sketch the main steps  to illustrate that all conditions needed are satisfied. The same notation and definitions as in the proof of Proposition \ref{proposition1} are considered.
First note that it is sufficient to show that  $|\logit \tilde\alpha_t (\hat \theta_T)-\logit \tilde\alpha_t^*|\xrightarrow{a.s.}0$ as $T\rightarrow \infty$. This because we have 
$$|\logit \hat \alpha_t (\hat \theta_T)-\logit \tilde\alpha_t^*|\le|\logit \tilde\alpha_t (\hat \theta_T)-\logit \tilde\alpha_t^*|+\|\logit \hat\alpha_t-\logit \tilde\alpha_t\|_\Theta,$$
 and $\|\logit \hat\alpha_t-\logit \tilde\alpha_t\|_\Theta\xrightarrow{a.s.}0$ from Proposition \ref{proposition1}.
From the results in Theorem 2 of \cite{Win2013} and the assumptions considered in Proposition \ref{proposition1},  we have that for any $\theta \in \Theta$  there exists a compact neighborhood $B(\theta)$ of $\theta$ such that the contraction condition holds uniformly, namely $E\log (\|\Lambda_t\|_{B(\theta)})<0$. Therefore, this is true also for the pseudo-true parameter $\theta^*\in \Theta$. As in the proof of Theorem 3 of \cite{Win2013},  repeated applications of the mean value theorem yield
$$\|\logit \tilde\alpha_t (\cdot)-\logit \tilde\alpha_t^*\|_{B(\theta^*)}\le\sum_{k=1}^\infty \prod_{i=1}^k\|\Lambda_{t-i}\|_{B(\theta^*)} \|\phi_{t-k}(\logit \tilde\alpha_{t-k}^* ,\cdot)-\logit \tilde\alpha_{t-k+1}^*\|_{B(\theta^*)}$$
for any $\theta\in B(\theta^*)$ with probability 1. The existence of the limit on the right hand side is obtained from Lemma 2.1 of \cite{SM2006} together with  the integrability condition $E\log^+\|\logit\tilde \alpha_t\|_{B(\theta^*)}$, implied by Lemma \ref{lemma2}, and $\prod_{i=1}^k\|\Lambda_{t-i}\|_{B(\theta^*)}  \xrightarrow{e.a.s.}0$ as $k \rightarrow \infty$, implied by the uniform contraction condition.
Finally, the desired result $|\logit \tilde\alpha_t (\hat \theta_T)-\logit \tilde\alpha_t^*|\xrightarrow{a.s.}0$  follows as in Theorem 3 of \cite{Win2013} taking into account that   the ML estimator $\hat \theta_T $ is strongly consistent by Theorem \ref{theorem1}.
\end{proof}

\begin{proof}[Proof of Theorem \ref{theorem2}]
An application of the mean value theorem  together with Lemma \ref{lemma4} yields that for any $x\in \mathbb{N}$ there is a $C_x>0$  and a   stationary sequence of random variables $\{\eta_t\}_{t\in \mathbb{N}}$ such that the following inequalities hold true with probability 1
\begin{align*}
|\hat p_t(x,\hat \theta_T)- p^*_t(x)|\le& \sup_{(\alpha,\theta)\in (0,1) \times \Theta}\left|\frac{\partial p(x|y_{t-1},\alpha,\xi)}{\partial \logit\alpha}\right| \left|\logit\hat \alpha_t(\hat \theta_T)-\logit\alpha_t^*\right|+ \\ 
+&\sup_{(\alpha,\theta)\in (0,1) \times \Theta}\left\|\frac{\partial p(x|y_{t-1},\alpha,\xi)}{\partial \xi}\right\|_1\|\hat\xi_T-\xi^*\|_1 \\
\le &  \eta_t| \logit   \hat\alpha_t(\hat \theta_T)-\logit \alpha_t^*| +C_x\|\hat\xi_T-\xi^*\|_1.
\end{align*}
The desired convergence to zero in probability of $|\hat p_t(x,\hat \theta_T)- p^*_t(x)|$  then follows immediately as $\|\hat\xi_T-\xi^*\|_1$ is $o_p(1)$ by  Theorem \ref{theorem1} and $|\logit \hat \alpha_t(\hat \theta_T)-\logit\alpha_t^*|$ is  $o_p(1)$ by Lemma \ref{lemma1.0}.
\end{proof}

\subsection{Technical lemmas}

\begin{lemma}
\label{lemma1}
Let Assumption \ref{assumption1} hold, then the following inequalities are satisfied with probability 1 for any $\alpha \in (0,1)$ and  $\xi \in \Xi$ 

\begin{description}
\item[(i)] $\left|s_t(\alpha, \xi)\right|\le 2y_{t-1}.$
\item[(ii)] $- y_{t-1}/4 \le \dot s_t(\alpha, \xi) \le m_t^2.$
\end{description}
\end{lemma}

\begin{proof}

Assumption \ref{assumption1} implies that $p_{kt}(\alpha,\xi)>0$ with probability 1 for any $\alpha\in (0,1)$ and $\xi \in \Xi$. This ensures that  $s_t(\alpha,\xi)$ and $\dot s_t(\alpha,\xi)$ are well defined as their denominator, see expressions (\ref{first_der}) and (\ref{second_der}), is almost surely larger then zero for any $\alpha\in (0,1)$ and $\xi \in \Xi$.

To show that (i) is satisfied, we note that 
$$\left|s_t(\alpha, \xi)\right| \le \left(\sum_{k=0}^{m_t}p_{kt}(\alpha,\xi)\right)^{-1}\left(\sum_{k=0}^{m_t} p_{kt}(\alpha,\xi)(k+y_{t-1}\alpha) \right)\le (1+\alpha)y_{t-1},$$
therefore (i) immediately holds true as $\alpha \in (0,1)$.

As concerns (ii), taking into account that $y_t\ge 0$ almost surely,~we obtain that the  numerator of expression (\ref{second_der}) is smaller or equal than
$$\left(\sum_{j=0}^{m_t}\sum_{k=0}^{m_t}p_{kt}(\alpha,\xi)p_{jt}(\alpha,\xi) k(k-j)\right)  \le \left(\sum_{j=0}^{m_t}\sum_{k=0}^{m_t}p_{kt}(\alpha,\xi)p_{jt}(\alpha,\xi)\right) m_t^2,$$
therefore it follows immediately that $\dot s_t(\alpha, \xi) \le m_t^2$. Similarly, we obtain that the numerator of (\ref{second_der}) is larger or equal than 
$$\left(\sum_{j=0}^{m_t}\sum_{k=0}^{m_t}p_{kt}(\alpha,\xi)p_{jt}(\alpha,\xi)\right) (-\alpha(1-\alpha)y_{t-1}),$$
therefore  $ \dot s_t(\alpha, \xi) \ge - y_{t-1}/4 $ as $\alpha \in (0,1)$ and, as a result,   it follows that (ii) is satisfied. 
\end{proof}

\begin{lemma}
\label{lemma2}
Let the conditions of Proposition \ref{proposition1} hold, then $E\|\logit \tilde \alpha_t(\theta)\|_{\Theta}^2<\infty$.
\end{lemma}

\begin{proof}
The lemma is proved by showing  that there exists a stationary and ergodic sequence  $\{\tilde\nu_t\}_{t \in \mathbb{Z}}$ such that $E\tilde \nu_t^2<\infty$  and that $\|\logit \tilde \alpha_t\|_\Theta<(\tilde \nu_t+1)$ with probability 1. Then, it is immediate to conclude that $E\|\logit \tilde \alpha_t\|_\Theta^2<\infty$.

First, we define the sequence $\{\hat v_t\}_{t\in \mathbb{N}}$ through the following stochastic recurrence equation
$$\hat v_{t+1} = \omega_u+\beta_u \hat v_t+2\tau_u y_{t},\; t\in \mathbb{N},$$
 which is initialized at $\hat v_0=\omega_u/(1-\beta_u)$ and where $\omega_u=\sup_{\theta\in\Theta}|\omega|$, $\beta_u=\sup_{\theta\in\Theta}|\beta|$ and $\tau_u=\sup_{\theta\in\Theta}|\tau|$. Considering that $\beta_u<1$ from the specification of $\Theta$ and that $\{y_t\}_{t\in\mathbb{Z}}$ is stationary and ergodic,  an application of Theorem 3.1 of \cite{Bougerol1993} yields that $|\hat v_{t} -\tilde v_t|\xrightarrow{a.s.}0$ as $t$ goes to infinity, where $\{\tilde v_t\}_{t \in \mathbb{N}}$ is a stationary and ergodic sequence that admits the following representation
 \begin{align*}
\tilde v_t=\omega_u/(1-\beta_u)+2\tau_u\sum_{k=1}^{\infty}\beta_u^k y_{t-k}.
 \end{align*}
 From this expression, it is straightforward to obtain that $Ey_{t}^2<\infty$, together with $\beta_u<1$, entails $E \tilde v_t^2<\infty$.

 In the following, we show that  $\|\logit \tilde \alpha_t\|_\Theta<(\tilde \nu_t+1)$ with probability 1. Without loss of generality we can  assume that the filter $\{\logit \hat \alpha_t(\theta)\}_{t \in \mathbb{N}}$ is initialized at $\hat \alpha_0(\theta)=\omega/(1-\beta)$.   Now, taking into account  that $\sup_{\theta \in \Theta}|s_t(\alpha,\xi)|<2 y_{t-1}$ a.s.~for any $\alpha\in (0,1)$ by Lemma \ref{lemma1}, it follows immediately that $\|\logit\hat\alpha_t\|_\Theta\le\hat v_t$ with probability 1 for any $t \in \mathbb{N}$. Therefore, we have that for a large enough $t\in\mathbb{N}$ with probability 1
 \begin{align*}
\|\logit \tilde \alpha_t \|_\Theta- \tilde v_t -1& \le \|\logit \hat \alpha_t \|_\Theta- \hat  v_t -1+\|\logit \tilde \alpha_t -\logit \hat \alpha_t \|_\Theta+|\tilde  v_t- \hat  v_t|<0,
  \end{align*}
as $\|\logit \tilde \alpha_t -\logit \hat \alpha_t \|_\Theta$ and $|\tilde  v_t- \hat  v_t|$ go to zero almost surely. As a result, given the stationarity of $\{\|\logit \tilde \alpha_t \|_\Theta- \tilde v_t\}_{t \in \mathbb{Z}}$ we infer that $\|\logit \tilde \alpha_t \|_\Theta<(\tilde v_t+1)$ with probability 1 for any $t \in \mathbb{Z}$. This concludes the proof.
\end{proof}

\begin{lemma}
\label{lemma4}
Let the conditions of Theorem \ref{theorem2} hold. Then, for any $x\in \mathbb{N}$ there exists a stationary sequence of random variables $\{\eta_t\}_{t\in\mathbb{N}}$ and a constant $C_x>0$ such that almost surely
\begin{description}
\item[(i)] $ \sup_{(\alpha,\theta)\in (0,1) \times \Theta}\left|\frac{\partial p(x|y_{t-1},\alpha,\xi)}{\partial \logit \alpha}\right|\le \eta_t.$
\item[(ii)] $\sup_{(\alpha,\theta)\in (0,1) \times \Theta}\left\|\frac{\partial p(x|y_{t-1},\alpha,\xi)}{\partial \xi}\right\|_1\le C_x.$
\end{description}
\end{lemma}

\begin{proof}
First we show that (i) holds true. From elementary calculus, we obtain that 
$$\frac{\partial p(x|y_{t-1},\alpha,\xi)}{\partial \logit \alpha}=\sum_{k=0}^{m_{xt}}p_{kt}(x,\alpha,\xi) (k-\alpha y_{t-1}),$$
where $m_{x t}=\min(x,y_{t-1})$ and 
\begin{align*}
p_{kt}(x,\alpha,\xi)= {y_{t-1}\choose k}  \alpha^k (1- \alpha)^{y_{t-1}-k}p_e(x-k, \xi).
\end{align*}
As a result, taking into account that $0\le p_{kt}(x,\alpha,\xi) \le 1$ with probability 1 for any $(x,\alpha,\xi)\in \mathbb{N}\times (0,1)\times \Xi$, it follows that
\begin{align*}
\left|\frac{\partial p(x|y_{t-1},\alpha,\xi)}{\partial \logit \alpha}\right|&\le \sum_{k=0}^{m_{xt}}p_{kt}(x,\alpha,\xi) (k+y_{t-1})\le \sum_{k=0}^{y_{t-1}}(k+y_{t-1})\le 2(1+y_{t-1})y_{t-1}.
\end{align*}
Therefore, the result (i) is proved setting $\eta_t=2(1+y_{t-1})y_{t-1}$ and recalling that $\{y_t\}_{t\in\mathbb{Z}}$ is stationary and ergodic and thus $\{\eta_t\}_{t\in\mathbb{Z}}$  is stationary and ergodic  as well.

As concerns (ii),  we have that
$$\frac{\partial p(x|y_{t-1},\alpha,\xi)}{\partial \xi}=\sum_{k=0}^{m_{xt}}{y_{t-1}\choose k}  \alpha^k (1- \alpha)^{y_{t-1}-k} \frac{\partial p_e(x-k, \xi)}{\partial \xi}.$$
As a result, we obtain that the following inequalities are satisfied almost surely
\begin{align*}
\left\|\frac{\partial p(x|y_{t-1},\alpha,\xi)}{\partial \logit \alpha}\right\|_1&\le \sum_{k=0}^{m_{xt}}{y_{t-1}\choose k}  \alpha^k (1- \alpha)^{y_{t-1}-k}\left\| \frac{\partial p_e(x-k, \xi)}{\partial \xi}\right\|_1\le\sum_{k=0}^{x}\left\| \frac{\partial p_e(x-k, \xi)}{\partial \xi}\right\|_1.
\end{align*}
Therefore, from the continuity of the  derivative provided by Assumption \ref{assumption2.5} and the compactness of $\Theta$, we obtain that for any given $x-k\in\mathbb{N}$ there is a constant $C_{kx}>0$ such that 
$$\sup_{\theta\in\Theta}\left\|\frac{\partial p_e(x-k, \xi)}{\partial \xi} \right\|_1\le C_{kx}.$$
This shows that the  result in (ii) holds as  $C_x=\sum_{k=0}^x C_{kx}<\infty$.
\end{proof}

\bibliographystyle{apalike}
\bibliography{references}

\begin{thebibliography}{}

\bibitem[Al-Osh and Aly, 1992]{al1992first}
Al-Osh, M.~A. and Aly, E.-E.~A. (1992).
\newblock First order autoregressive time series with negative binomial and
  geometric marginals.
\newblock {\em Communications in Statistics-Theory and Methods},
  21(9):2483--2492.

\bibitem[Al-Osh and Alzaid, 1987]{ALOSH1987}
Al-Osh, M.~A. and Alzaid, A.~A. (1987).
\newblock First-order integer valued autoregressive (inar(1)) process.
\newblock {\em Journal of Time Series Analysis}, 8(3):261--275.

\bibitem[Alzaid and Al-Osh, 1990]{alzaid1990}
Alzaid, A. and Al-Osh, M. (1990).
\newblock An integer-valued pth-order autoregressive structure (inar (p))
  process.
\newblock {\em Journal of Applied Probability}, 27(2):314--324.

\bibitem[Blasques et~al., 2016]{cb2016}
Blasques, F., Koopman, S.~J., Lasak, K., and Lucas, A. (2016).
\newblock In-sample confidence bands and out-of-sample forecast bands for
  time-varying parameters in observation-driven models.
\newblock {\em International Journal of Forecasting}, 32(3):875--887.

\bibitem[Blasques et~al., 2015]{blasq2015}
Blasques, F., Koopman, S.~J., and Lucas, A. (2015).
\newblock Information-theoretic optimality of observation-driven time series
  models for continuous responses.
\newblock {\em Biometrika}, 102(2):325--343.

\bibitem[Bollerslev, 1986]{bol1986}
Bollerslev, T. (1986).
\newblock Generalized autoregressive conditional heteroskedasticity.
\newblock {\em Journal of Econometrics}, 31(3):307--327.

\bibitem[Bougerol, 1993]{Bougerol1993}
Bougerol, P. (1993).
\newblock Kalman filtering with random coefficients and contractions.
\newblock {\em SIAM Journal on Control and Optimization}, 31(4):942--959.

\bibitem[Creal et~al., 2011]{creal2012dynamic}
Creal, D., Koopman, S.~J., and Lucas, A. (2011).
\newblock A dynamic multivariate heavy-tailed model for time-varying
  volatilities and correlations.
\newblock {\em Journal of Business \& Economic Statistics}, 29(4):552--563.

\bibitem[Creal et~al., 2013]{Creal2013}
Creal, D., Koopman, S.~J., and Lucas, A. (2013).
\newblock Generalized autoregressive score models with applications.
\newblock {\em Journal of Applied Econometrics}, 28(5):777--795.

\bibitem[Davis et~al., 2003]{Davis2003}
Davis, R.~A., Dunsmuir, W. T.~M., and Streett, S.~B. (2003).
\newblock Observation‐driven models for poisson counts.
\newblock {\em Biometrika}, 90(4):777--790.

\bibitem[Engle, 1982]{engle1982}
Engle, R.~F. (1982).
\newblock Autoregressive conditional heteroscedasticity with estimates of the
  variance of united kingdom inflation.
\newblock {\em Econometrica}, 50(4):987--1007.

\bibitem[Freeland and McCabe, 2004]{Freeland2004427}
Freeland, R. and McCabe, B. (2004).
\newblock Forecasting discrete valued low count time series.
\newblock {\em International Journal of Forecasting}, 20(3):427 -- 434.

\bibitem[Harvey, 2013]{H2013}
Harvey, A. (2013).
\newblock {\em Dynamic Models for Volatility and Heavy Tails: With Applications
  to Financial and Economic Time Series}.
\newblock New York: Cambridge University Press.

\bibitem[Harvey and Luati, 2014]{HL2014}
Harvey, A. and Luati, A. (2014).
\newblock Filtering with heavy tails.
\newblock {\em Journal of the American Statistical Association},
  109(507):1112--1122.

\bibitem[Jazi et~al., 2012]{jazi2012}
Jazi, M.~A., Jones, G., and Lai, C.-D. (2012).
\newblock First-order integer valued ar processes with zero inflated poisson
  innovations.
\newblock {\em Journal of Time Series Analysis}, 33(6):954--963.

\bibitem[Jin-Guan and Yuan, 1991]{jin1991}
Jin-Guan, D. and Yuan, L. (1991).
\newblock The integer-valued autoregressive (inar (p)) model.
\newblock {\em Journal of time series analysis}, 12(2):129--142.

\bibitem[Kim and Park, 2008]{kim2008non}
Kim, H.-Y. and Park, Y. (2008).
\newblock A non-stationary integer-valued autoregressive model.
\newblock {\em Statistical papers}, 49(3):485--502.

\bibitem[Krengel, 1985]{krengel1985}
Krengel, U. (1985).
\newblock {\em Ergodic theorems}.
\newblock de Gruyter, Berlin.

\bibitem[McKenzie, 1988]{MCK1988}
McKenzie, E. (1988).
\newblock Some arma models for dependent sequences of poisson counts.
\newblock {\em Advances in Applied Probability}, 20:822--835.

\bibitem[Pedeli and Karlis, 2011]{pedeli2011bivariate}
Pedeli, X. and Karlis, D. (2011).
\newblock A bivariate inar (1) process with application.
\newblock {\em Statistical modelling}, 11(4):325--349.

\bibitem[Rao, 1962]{rao1962}
Rao, R.~R. (1962).
\newblock Relations between weak and uniform convergence of measures with
  applications.
\newblock {\em The Annals of Mathematical Statistics}, 33(2):659--680.

\bibitem[Salvatierra and Patton, 2015]{salvatierra2015dynamic}
Salvatierra, I. D.~L. and Patton, A.~J. (2015).
\newblock Dynamic copula models and high frequency data.
\newblock {\em Journal of Empirical Finance}, 30:120--135.

\bibitem[Steutel and Van~Harn, 1979]{steutel1979discrete}
Steutel, F. and Van~Harn, K. (1979).
\newblock Discrete analogues of self-decomposability and stability.
\newblock {\em The Annals of Probability}, 7(5):893--899.

\bibitem[Straumann and Mikosch, 2006]{SM2006}
Straumann, D. and Mikosch, T. (2006).
\newblock Quasi-maximum-likelihood estimation in conditionally heteroscedastic
  time series: A stochastic recurrence equations approach.
\newblock {\em The Annals of Statistics}, 34(5):2449--2495.

\bibitem[Wald, 1949]{wald1949}
Wald, A. (1949).
\newblock Note on the consistency of the maximum likelihood estimate.
\newblock {\em The Annals of Mathematical Statistics}, 20(4):595--601.

\bibitem[White, 1982]{white1982}
White, H. (1982).
\newblock Maximum likelihood estimation of misspecified models.
\newblock {\em Econometrica}, 50(1):1--25.

\bibitem[Wintenberger, 2013]{Win2013}
Wintenberger, O. (2013).
\newblock Continuous invertibility and stable qml estimation of the egarch(1,1)
  model.
\newblock {\em Scandinavian Journal of Statistics}, 40(4):846--867.

\bibitem[Zheng and Basawa, 2008]{zheng2008first}
Zheng, H. and Basawa, I.~V. (2008).
\newblock First-order observation-driven integer-valued autoregressive
  processes.
\newblock {\em Statistics \& Probability Letters}, 78(1):1--9.

\bibitem[Zheng et~al., 2007]{Zheng2007212}
Zheng, H., Basawa, I.~V., and Datta, S. (2007).
\newblock First-order random coefficient integer-valued autoregressive
  processes.
\newblock {\em Journal of Statistical Planning and Inference}, 137(1):212 --
  229.

\end{thebibliography}

\end{document}